\documentclass[oneside,english]{amsart}
\usepackage[T1]{fontenc}
\usepackage[latin9]{inputenc}
\usepackage{geometry}
\usepackage{color}
\usepackage{amstext}
\usepackage{amsthm}
\usepackage{amssymb}
\usepackage{graphicx}
\usepackage{subfigure}
\usepackage{amsaddr}

\makeatletter
\numberwithin{equation}{section}
\numberwithin{figure}{section}
\theoremstyle{plain}
\newtheorem{theorem}{\protect\theoremname}
  \theoremstyle{definition}
  \newtheorem{definition}[theorem]{\protect\definitionname}
  \theoremstyle{plain}
  \newtheorem{lemma}[theorem]{\protect\lemmaname}

\usepackage{babel}

\makeatother

\usepackage{babel}
  \providecommand{\definitionname}{Definition}
  \providecommand{\lemmaname}{Lemma}
\providecommand{\theoremname}{Theorem}

\begin{document}

\title{Long-range Interactions and Network Synchronization}

\author{Ernesto Estrada}%
\address{Department of Mathematics and Statistics, University of Strathclyde, 26 Richmond Street,\\ Glasgow G1 1HX, United Kingdom}

\author{Lucia Valentina Gambuzza}
\address{Department of Electrical, Electronic and Computer Engineering, University of Catania, \\ Viale A. Doria 6, Catania, Italy}

\author{Mattia Frasca}
\address{Department of Electrical, Electronic and Computer Engineering, University of Catania, \\ Viale A. Doria 6, Catania, Italy}

\maketitle

\begin{abstract}
The dynamical behavior of networked complex systems is shaped not only by the direct links among the units, but also by the long-range interactions occurring through the many existing paths connecting the network nodes. In this work, we study how synchronization dynamics is influenced by these long-range interactions, formulating a model of coupled oscillators that incorporates this type of interactions through the use of $d-$path Laplacian matrices. We study synchronizability of these networks by the analysis of the Laplacian spectra, both theoretically and numerically, for real-world networks and artificial models. Our analysis reveals that in all networks long-range interactions improve network synchronizability with an impact that depends on the original structure, for instance it is greater for graphs having a larger diameter. We also investigate the effects of edge removal in graphs with long-range interactions and, as a major result, find that the removal process becomes more critical, since also the long-range influence of the removed link disappears.
\end{abstract}

%

\section{Introduction}

The fact that most of complex systems are networked has made complex networks an important paradigm for studying such systems \cite{newman2003structure,boccaletti2006complex,estrada2012structure}. A complex network is a graph that represents the skeleton of such complex systems, ranging from biological and ecological to social and infrastructural ones. Representing such a variety of systems by a single mathematical object can be considered as a drastic simplification. However, complex networks have been very useful in explaining many properties of complex systems, which has been empirically validated their use for this purpose.
In order to fill the gaps left by this simplified representation of complex systems a few extensions beyond the simple graph have been proposed. They include the use of hypergraphs, multiplexes and multilayer networks \cite{de2013mathematical,kivela2014multilayer}, temporal graphs \cite{holme2012temporal} and more recently the use of simplicial complexes \cite{courtney2016generalized,estrada2017centralities}.

The previously mentioned extensions of complex systems representations try to ameliorate the emphasis paid by graphs on binary relations only. Then, in either of the previously developed representations of complex systems\textemdash hypergraphs, multiplexes, simplicial complexes\textemdash the binary relation is replaced by a unified $k$-ary one, in which the individual is replaced by the group. Just to illustrate one example, in the hypergraph the binary relation between individuals is replaced by the $k$-ary hyperedges in which nodes are assumed to be identical in their connectivity inside this group. Then, an important missing aspect of these representations is how to capture the influence of nodes in a network in a way that gradually decays with the separation of these nodes in the graph.

Recently, such approach has been proposed to account for long-range influences in a network by using extensions of the graph-theoretic concepts of adjacency and connectivity \cite{estrada2013peer,estrada2017path}. In these works, a new paradigm is developed in which a node in a network is not only influenced by its nearest neighbours but by any other node in the graph. However, such influence is gradually tuned by the shortest path distance at which the influencers are from the influencee. It is obvious that this type of representation is not general for any kind of complex systems, as there are cases where such long-range influence does not exists. However, in the case of social systems such kind of long-range influence is certainly manifested in the so-called indirect peers pressure. Individuals in a social group are not only directly influenced by those connected to them but also by those socially close to them.

The approach proposed in \cite{estrada2013peer,estrada2017path} is here applied to study how the long-range influences affect synchronization of the network nodes. In fact, when the units of a network are dynamical systems (for instance, periodic or chaotic oscillators), a collective phenomenon, characterized by the emergence of a common rhythm in all the units and observed in many natural and artificial systems, may emerge as the result of the interactions \cite{strogatz2004sync}. In this context it is known that the topology of the connections plays a fundamental role in determining the characteristics of the synchronous motion, its onset and stability \cite{arenas2008synchronization}. Most of the works on the subject, however, consider interactions as only dictated by direct links, whereas, in this paper, we take into account influences through paths of length grater than one. Other studies \cite{rogers1996phase,marodi2002synchronization,anteneodo2003analytical,li2006phase} have considered the case of oscillators embedded on a geometrical space and coupled with an intensity decaying with the geometric distance between them. In our paper, instead, the oscillators are viewed as the nodes of a network and the intensity of the interactions is weighted by considering the distance of the oscillators as measured by the length of the shortest path connecting them. This scenario has been modelled through the use of $d-$path Laplacian matrices, whose spectra are shown to determine the synchronizability of the system. In particular, we prove how increasingly weighting the long-range interactions always leads to the best scenario possible for synchronization and verify the result on network models and real-world structures. The quantitative impact of the long-range interactions depends on the network topology and on its specific properties such as the diameter, the density and the degree distribution. Finally, we study the effect of edge removal in networks with and without long-range interactions and, as a major finding, observe that it is more critical in the presence of such interactions, still exhibiting a higher synchronizability in this case than when the long-range interactions are not present.

\section{Intuition and Mathematical formulation}

In this section we formulate the general mathematical equations for considering long-range interactions (LRIs) in a system of coupled dynamical units. Let $G=\left(V,E\right)$ be a simple, undirected graph without self-loops having $N$ nodes and $m$ edges. Let us now write the equations of $N$ coupled dynamical units on a graph. Let us consider $N$ identical oscillators coupled with a coupling constant $\sigma$, where the oscillator $i$ has state variables $\mathbf{x}_{i}\in\mathcal{R}^{n}$.
Then,

\begin{equation}
\dot{\mathbf{x}}_{i}=f(\mathbf{x}_{i})+\sigma\sum_{\left(i,j\right)\in E}\left(H(\mathbf{x}_{j})-H(\mathbf{x}_{i})\right),\label{eq:coupled}
\end{equation}

\noindent where $f(\mathbf{x}_{i}):\mathcal{R}^{n}\rightarrow \mathcal{R}^{n}$ represents the uncoupled dynamics, and $H(\mathbf{x}_{j}):\mathcal{R}^{n}\rightarrow \mathcal{R}^{n}$ the coupling function.

In matrix form Eqs. (\ref{eq:coupled}) can be written as follow

\begin{equation}
\dot{\vec{\mathbf{x}}}=\vec{f}(\vec{\mathbf{x}})-\sigma\left(\nabla\cdot \nabla^{T}\otimes \mathrm{I}_n\right)\cdot \vec{H}(\vec{\mathbf{x}}),
\end{equation}
where $\vec{\mathbf{x}}=\left[\mathbf{x}_{1}~\mathbf{x}_{2}\ldots\mathbf{x}_{N}\right]^{T}$, $\vec{f}(\vec{\mathbf{x}})=\left[f(\mathbf{x}_{1})~f(\mathbf{x}_{2})\ldots f(\mathbf{x}_{N})\right]^{T}$, $\mathrm{I}_n$ indicates the identity matrix of order $n$, and $\vec{H}(\vec{\mathbf{x}})=\left[H(\mathbf{x}_{1})~H(\mathbf{x}_{2})\ldots H(\mathbf{x}_{N})\right]^{T}$. The entries of the node-to-edges incidence matrix $\nabla\in\mathbb{R}^{N\times m}$
of the graph are defined as
\noindent \begin{flushleft}
\begin{equation}
\nabla_{ij}=\left\{ \begin{array}{r}
1\\
-1\\
0
\end{array}\right.\begin{array}{l}
\textnormal{if node \ensuremath{i} is the head of the edge \ensuremath{j}},\\
\textnormal{if node \ensuremath{i} is the tail of edge \ensuremath{j,}}\\
\textnormal{otherwise.}
\end{array}
\end{equation}
\par\end{flushleft}

We recall that the matrix $L=\nabla\nabla^{T}$ is known as the Laplacian
matrix of the graph and has entries
\noindent \begin{flushleft}
\begin{equation}
L_{ij}=\left\{ \begin{array}{r}
k_{i}\\
-1\\
0
\end{array}\right.\begin{array}{l}
\textnormal{if \ensuremath{i=j} },\\
\textnormal{if \ensuremath{\left(i,j\right)\in E}, }\\
\textnormal{otherwise,}
\end{array}
\end{equation}
where $k_{i}$ is the degree of the node $i$, the number of nodes
adjacent to it. It is worth noting that the Laplacian matrix is related
to the adjacency matrix of the graph via: $L=K-A$, where $K$ is
the diagonal matrix of degrees and the adjacency matrix has entries
\par\end{flushleft}

\begin{equation}
A_{ij}=\left\{ \begin{array}{c}
1\\
0
\end{array}\right.\begin{array}{c}
\textnormal{if}\ \left(i,j\right)\in E,\\
\textnormal{otherwise}.
\end{array}
\end{equation}

Now, let us consider the following scenario of $N$ oscillators which
are coupled according to a graph $G$ with a coupling strength $\sigma_{1}$.
Consider that the oscillators can also couple in a weaker way if they
are separated at a shortest path distance two. Let $\sigma_{2}$ be
the strength of the coupling between those oscillators, such that
$\sigma_{2}<\sigma_{1}$. In a similar way we can consider that oscillators
at a given shortest path distance can couple together with a coupling
strength which depends on their separation in the network. For instance,
pairs of oscillators at distance three can couple with strength $\sigma_{3}<\sigma_{2}<\sigma_{1}$.
This situation may well represents social scenarios as illustrated
in Figure \ref{levels of influence}. In a social network, individuals
are connected by certain social ties, such as friendship, collaboration,
etc. Then, two individual directly connected to each other can influence
each other in a relatively strong way. However, an individual in this
network can also receive certain influence from others which are not
directly connected to her. This influence is supposed to be smaller
than the ones received from direct acquaintances, but not one that
can be discarded at all. Let us consider a simple example of a scientific
collaboration network. Two individuals are connected in this network
if they have collaborated on a certain topic, e.g., they have published
a paper together. It is clear that they have influenced each other
in terms of their scientific styles. However, these two scientists
are also influenced by individuals which are close to their topic
of research although they have not collaborated together. This closeness
is reflected in the fact that they are relatively close in this collaboration
network in terms of shortest path distance. An individual from a completely
different field is expected to be far from them in terms of the shortest path
distance, and so to have a lower influence in their scientific styles.

\begin{figure}
\begin{centering}
\includegraphics[width=0.45\textwidth]{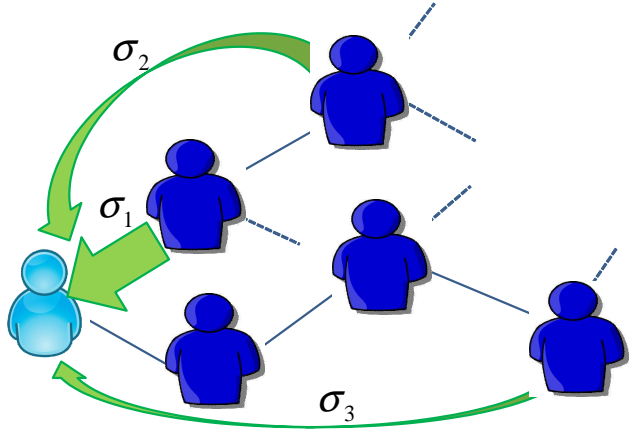}
\par\end{centering}
\caption{Schematic representation of long-range influences in a social group.
It is assumed that $\sigma_{d_{max}}<\cdots<\sigma_{3}<\sigma_{2}<\sigma_{1},$
where $d_{max}$ is the diameter of the graph. }
\label{levels of influence}
\end{figure}

In mathematical terms this intuition can be formulated as follows.
Let $\sigma_{1},\sigma_{2},\ldots,\sigma_{d_{max}}$ account for the
strength of the cupling between pairs of oscillators separated by
shortest path distances of one, two, and so forth up to the diameter
of the graph, $d_{max}$. Then, we can write

\begin{equation}
\begin{array}{lll}
\dot{\mathbf{x}}_{i}  & = & f(\mathbf{x}_{i})+\sigma_{1}\sum\limits_{\left(i,j\right)\in E}\left(H(\mathbf{x}_{j})-H(\mathbf{x}_{i})\right)+\sigma_{2}\sum\limits_{d\left(i,k\right)=2}\left(H(\mathbf{x}_{k})-H(\mathbf{x}_{i})\right)+\cdots\\
& & \cdots+ \sigma_{d_{max}}\sum\limits_{d\left(i,r\right)=d_{max}}\left(H(\mathbf{x}_{r})-H(\mathbf{x}_{i})\right).
\end{array}
\end{equation}


The simplest way to account for these long-range influences is by
considering that the coupling strength between individuals decays as
certain law of the shortest path distance. That is, if we consider
a power-law decay of the strength of coupling with the distance we
get, assuming that $\sigma_{1}=\sigma$ and that the rest are just
a fraction of it,

\begin{equation}
\begin{array}{lll}
\dot{\mathbf{x}}_{i} & = & f(\mathbf{x}_{i})+\sigma\sum\limits_{\left(i,j\right)\in E}\left(H(\mathbf{x}_{j})-H(\mathbf{x}_{i})\right)+\sigma2^{-s}\sum\limits_{d\left(i,k\right)=2}\left(H(\mathbf{x}_{k})-H(\mathbf{x}_{i})\right)+\cdots\\ & & \cdots +\sigma d_{max}^{-s}\sum\limits_{d\left(i,r\right)=d_{max}}\left(H(\mathbf{x}_{r})-H(\mathbf{x}_{i})\right).
\end{array}
\end{equation}

Similarly, if we assume an exponential decay we obtain

\begin{equation}
\begin{array}{lll}
\dot{\mathbf{x}}_{i}& = & f(\mathbf{x}_{i})+\sigma\sum\limits_{\left(i,j\right)\in E}\left(H(\mathbf{x}_{j})-H(\mathbf{x}_{i})\right)+\sigma e^{-2\lambda}\sum\limits_{d\left(i,k\right)=2}\left(H(\mathbf{x}_{k})-H(\mathbf{x}_{i})\right)+\cdots\\ & & \cdots+\sigma e^{-d_{max}\lambda}\sum\limits_{d\left(i,r\right)=d_{max}}\left(H(\mathbf{x}_{r})-H(\mathbf{x}_{i})\right).
\end{array}
\end{equation}

We call these equations the Mellin and Laplace transforms, respectively,
of the $N$ coupled dynamical units on a graph.

Let us now define the $d$-path incidence matrices which account for
these coupling of non-nearest-neighbours in the graph. Let $P_{l,ij}$
denote a shortest path of length $l$ between $i$ and $j$. The nodes
$i$ and $j$ are called the endpoints of the path $P_{l,ij}$.
Because there could be more than one shortest path of length $l$
between the nodes $i$ and $j$ we introduce the following concept.
The irreducible set of shortest paths of length $l$ in the graph
is the set $P_{l}=\left\{ P_{l,ij},P_{l,ir},...,P_{l,st}\right\} $
in which the endpoints of every shortest path in the set are different.
Every shortest path in this set is called an irreducible shortest path.
Let $d_{max}$ be the graph diameter, i.e., the maximum shortest path
distance in the graph.
\begin{definition}
Let $d\leq d_{max}$. The $d$-path incidence matrix, denoted by $\nabla_{d}\in\mathbb{R}^{n\times p}$,
of a connected graph of $N$ nodes and $p$ irreducible shortest paths
of length $d$ is defined as:

\begin{equation}
\nabla_{d,ij}=\left\{ \begin{array}{r}
1\\
-1\\
0
\end{array}\right.\begin{array}{l}
\textnormal{if node \ensuremath{i} is the head of the irreducible shortest path \ensuremath{j}},\\
\textnormal{if node \ensuremath{i} is the tail of the irreducible shortest path \ensuremath{j}},\\
\textnormal{otherwise.}
\end{array}
\end{equation}
\end{definition}

Obviously $\nabla_{1}=\nabla.$ Let us now rewrite our Mellin and
Laplace transformed equations, respectively, in matrix-vector form using the $d$-path
incidence matrix as follow

\begin{equation}
\dot{\vec{\mathbf{x}}}=\vec{f}(\vec{\mathbf{x}})-\sigma\left(\sum_{d=1}^{d_{max}}d^{-s}\left(\nabla_{d}\cdot\nabla_{d}^{T}\otimes \mathrm{I}_n \right)\cdot\vec{H}(\vec{\mathbf{x}})\right),
\end{equation}

\begin{equation}
\dot{\vec{\mathbf{x}}}=\vec{f}(\vec{\mathbf{x}})-\sigma\left(\nabla\cdot\nabla^{T}\otimes \mathrm{I}_n \right)\cdot\vec{H}(\vec{\mathbf{x}})-\sigma\left(\sum_{d=2}^{d_{max}} e^{-\lambda d}\left(\nabla_{d}\cdot\nabla_{d}^{T}\otimes \mathrm{I}_n\right)\cdot \vec{H}(\vec{\mathbf{x}})\right).\label{eq:Laplacetrasformed}
\end{equation}

\noindent The parameters $s$ and $\lambda$ account for the strength of
the coupling of the oscillators at a given distance. The smallest
the values of these parameters the stronger the coupling between the
oscillators at a given distance. For instance when $s\rightarrow\infty$
($\lambda\rightarrow\infty$) there is a very weak influence of non-nearest
neighbours and we recover the classical model in which there is no
long-range coupling. When $s\rightarrow0$ ($\lambda\rightarrow0$)
the strength of the coupling between oscillators at any distance is
the same, which corresponds to the situation of the classical model
of coupled oscillators on a complete graph $K_{n}$. Thus, in every
case we always recover the original model of coupled oscillators on
graphs for large values of the parameters in the transforms of the
$d$-path incidence matrices and we approach the coupling on a complete
graph when these parameters tend to zero.

Note that in our approach the coupling strength depends on the shortest path distance. This differs from the notion of accessability matrix \cite{lentz2013unfolding}, which accounts for the existence of paths of arbitrary length between the nodes with unitary weights.

\subsection{Example. Long-range interactions in the Kuramoto model}

In this section we particularize the equations for the coupled system
to the case of phase oscillators, so that to consider LRIs in the Kuramoto model \cite{rodrigues2016kuramoto}. Let us consider $N$ phase oscillators on graph $G=\left(V,E\right)$ coupled with
an identical coupling constant $\sigma$, where the oscillator $i$
has phase $\theta_{i}$ and intrinsic frequency $\omega_{i}$. Then,

\begin{equation}
\dot{\theta}_{i}=\omega_{i}+\sigma\sum_{\left(i,j\right)\in E}\sin\left(\theta_{j}-\theta_{i}\right),
\end{equation}

\noindent or in matrix form

\begin{equation}
\dot{\vec{\theta}}=\vec{\omega}-\sigma\nabla\cdot\sin\left(\nabla^{T}\vec{\theta}\right).
\end{equation}

The consideration of LRIs in the way we have described previously will
give rise to the following transforms of the Kuramoto model:

\begin{equation}
\dot{\theta}_{i}  =  \omega_{i}+\sigma\sum_{\left(i,j\right)\in E}\sin\left(\theta_{j}-\theta_{i}\right)+\sigma2^{-s}\sum_{d\left(i,k\right)=2}\sin\left(\theta_{k}-\theta_{i}\right)+\cdots +\sigma d_{max}^{-s}\sum_{d\left(i,r\right)=d_{max}}\sin\left(\theta_{r}-\theta_{i}\right),
\end{equation}

\noindent for the Mellin transform, and

\begin{equation}
\begin{array}{lll}
\dot{\theta}_{i}& = & \omega_{i}+\sigma\sum\limits_{\left(i,j\right)\in E}\sin\left(\theta_{j}-\theta_{i}\right)+\sigma e^{-2\lambda}\sum\limits_{d\left(i,k\right)=2}\sin\left(\theta_{k}-\theta_{i}\right)+\cdots\\ & & \cdots+\sigma e^{-d_{max}\lambda}\sum\limits_{d\left(i,r\right)=d_{max}}\sin\left(\theta_{r}-\theta_{i}\right),
\end{array}
\end{equation}

\noindent for the Laplace transform. In matrix-vector form they are given by

\begin{equation}
\dot{\vec{\theta}}=\vec{\omega}-\sigma\left(\sum_{d=1}^{d_{max}}d^{-s}\nabla_{d}\cdot\sin\left(\nabla_{d}^{T}\vec{\theta}\right)\right),\label{eq:Mellin}
\end{equation}

\noindent for the Mellin transform, and

\begin{equation}
\dot{\vec{\theta}}=\vec{\omega}-\sigma\left(\nabla\cdot\sin\left(\nabla^{T}\vec{\theta}\right)\right)-\sigma\left(\sum_{d=2}^{d_{max}} e^{-\lambda d}\nabla_{d}\cdot\sin\left(\nabla_{d}^{T}\vec{\theta}\right)\right).\label{eq:Laplace}
\end{equation}
\noindent for the Laplace transform. Here again when $s\rightarrow\infty$ ($\lambda\rightarrow\infty$)
the coupling between oscillators at distance larger than one is almost
zero and we recover the classical Kuramoto model where there is no
LRIs. On the other hand, when $s\rightarrow0$ ($\lambda\rightarrow0$)
the strength of the coupling between oscillators at any distance is
the same and we obtain the Kuramoto model on a complete graph $K_{N}$.

\section{Synchronizability and Laplacian Spectra}

The Laplacian matrix of the graph is positive semi-definite with eigenvalues
denoted by: $0=\lambda_{1}\leq\lambda_{2}\leq\cdots\leq\lambda_{N}$.
If the network is connected, the multiplicity of the zero eigenvalue
is equal to one, i.e., $0=\lambda_{1}<\lambda_{2}\leq\cdots\leq\lambda_{N}$,
and the smallest nontrivial eigenvalue $\lambda_{2}$ is known as
the \textit{algebraic connectivity} of the network. It is now well
known that there are two types of networks with bounded and unbounded
synchronization regions in the parameter space. One large class of
dynamic networks have an unbounded synchronized region specified by
\begin{equation}
\sigma\lambda_{2}>\alpha_{1}>0,\label{unbounded}
\end{equation}
where constant $\alpha_{1}$ depends only on the node dynamics, a
bigger spectral gap $\lambda_{2}$ implies a better network synchronizability,
namely a smaller coupling strength $\sigma>0$ is needed~\cite{chen2008network,duan2007complex,duan2008network,lu2004characterizing}.

Another large class of dynamic networks have a bounded synchronized
region specified by
\begin{equation}
\sigma\lambda_{2},\ldots,\sigma\lambda_{N}\in(\alpha_{2},\alpha_{3})\subset(0,\infty),\label{bounded}
\end{equation}
where constants $\alpha_{2},\,\alpha_{3}$ depend only on the node
dynamics as well, and a bigger eigenratio $\lambda_{2}/\lambda_{n}$
implies a better network synchronizability, which likewise means a
smaller coupling strength is needed~\cite{pecora1998master,huang2009generic}.

Here we only consider the latter criterion while the former can be
discussed similarly. As introduced, in this scenario the synchronizability
of the graph depends on the eigenratio of the second smallest and the
largest eigenvalues of the Laplacian matrix:
\begin{equation}
Q=\dfrac{\lambda_{2}}{\lambda_{N}}.\label{eigenratio}
\end{equation}

Now, we have generalized the Laplacian matrix to the so-called $d$-path
Laplacian matrices which are defined as follow.
\begin{definition}
Let $d\leq d_{max}$. The $d$-path Laplacian matrix, denoted by $L_{d}\in\mathbb{R}^{n\times n}$,
of a connected graph of $n$ nodes is defined as:

\begin{equation}
L_{d,ij}=\left\{ \begin{array}{r}
\delta_{k}\left(i\right)\\
-1\\
0
\end{array}\right.\begin{array}{l}
\textnormal{if \ensuremath{i=j} },\\
\textnormal{if \ensuremath{d_{ij}=d} }\\
\textnormal{otherwise,}
\end{array}
\end{equation}
where $\delta_{k}\left(i\right)$ is the number of irreducible shortest paths
of length $d$ that are originated at node $i$, i.e., its $d$-path
degree. It is straightforward to realize that $L_{d}=\nabla_{d}\nabla_{d}^{T}$
in a similar way as it happens for the graph Laplacian.
\end{definition}
We now extend the notion of a connected component of a graph to the
$d$-connected component.
\begin{definition}
Let $G=\left(V,E\right)$ be a simple graph and let $d\leq d_{max}$.
A $d$-connected component of $G$ is a subgraph $G'=\left(V',E'\right)$,
$V'\subseteq V$, $E'\subseteq E$, such that for all $p,q\in V'$
there is a shortest path of length $d\left(p,q\right)=d$. If the
$d$-connected component includes all the vertices and edges of $G$,
the graph is said to be $d$-connected.
\end{definition}

The following is an important property of the $d$-path Laplacian
matrix proved by Estrada in 2012 \cite{estrada2012path}.

\begin{theorem}
Let $G$ be a simple graph and let $d\leq d_{max}$. The matrix \textup{$L_{d}$}
is positive semidefinite. Moreover, the multiplicity of the zero eigenvalue
of $L_{d}$ is equal to the number of $d$-connected components of
the graph.
\end{theorem}
The previous result has an important implication for the study of
synchronization on graphs using the $d$-path Laplacian matrix. Although
a graph can be connected in the usual way\textemdash here it should
be said 1-connected\textemdash it not necessarily should be $d$-connected.
Then, a synchronization process taking place between agents separated
at distance $d$ in that graph may not give rise to a unique synchronized
state. For instance, a synchronization process taking place on the
pairs of nodes separated at distance 2 in the pentagon $C_{5}$ gives
rise to a unique synchronized state because the graph is $2$-connected,
but a similar process on the hexagon $C_{6}$ conduces to two synchronized
states because there are two $2$-connected components in this graph.
Then, the use of the Mellin and Laplace transformed dynamical systems,
like the ones introduced in this work, require a transformation of
the $d$-path Laplacian matrices such that they account for the decay
of the influence of oscillators separated at different distances in
the graph. Then, we have the following.
\begin{definition}
Let $G=\left(V,E\right)$ be a simple connected graph and let $d\leq d_{max}$.
The Mellin and Laplace transformed $d$-path Laplacian matrices of
$G$ are given by
\end{definition}
\noindent \begin{flushleft}
\begin{equation}
\tilde{L_{\tau}}=\left\{ \begin{array}{l}
\sum_{d=1}^{d_{max}}d^{-s}L_{d},\\
L+\sum_{d=2}^{d_{max}} e^{-\lambda d}L_{d},
\end{array}\begin{array}{l}
\textnormal{for~}\tau=\textnormal{Mell},s>0\\
\textnormal{for~}\tau=\textnormal{Lapl},\lambda>0.
\end{array}\right.\label{eq:generalized_operator}
\end{equation}
\par\end{flushleft}

Now, we can interpret the LRI-model in the following way. Consider a simple connected graph $G=\left(V,E\right)$ and the following transformation: $f:G\left(V,E\right)\rightarrow G'\left(V,E',\phi,W\right)$, such that
$E'=\left\{ E\cup\left(p,q\right)\left|p,q\in V,\left(p,q\right)\notin E\right.\right\} $
and $\phi:W\rightarrow E'$ is a surjective mapping that assigns a weight to each of the elements of $E'$. The weights $w_{ij}\in W$ are given by the Mellin or Laplace transforms and they are specific for each graph, that is $w_{ij}=d_{ij}^{-s}$ or $w_{ij}=e^{-\lambda d_{ij}}$ for $d_{ij}>1$, and $w_{ij}=1$ for connected pairs of nodes. The main consequence of this is that we can generalize the synchronizability definition (\ref{eigenratio}) to:

\begin{equation}
Q_{\tau}=\dfrac{\lambda_{2}\left(\tilde{L}_{\tau}\right)}{\lambda_{n}\left(\tilde{L}_{\tau}\right)},
\end{equation}
where $\lambda_{n}\left(\tilde{L}_{\tau}\right)$ and $\lambda_{2}\left(\tilde{L}_{\tau}\right)$
are the largest and second smallest eigenvalues of the transformed $d$-path Laplacian matrices. The most important consequence of this new formulation is that when $s,\lambda\rightarrow\infty$ the weighted graph $G'\left(V,E',\phi,W\right)$ tends to the original graph $G=\left(V,E\right)$.
In addition, when $s,\lambda\rightarrow0$ the weighted graph $G'\left(V,E',\phi,W\right)$ tends to the complete graph with $N$ nodes, $K_{N}$. It is known that the Laplacian eigenvalues of $K_{N}$ are all equal to $N$ but one which is equal to zero. Then, the following result holds.

\begin{lemma}
\label{lemma1} Let $G=\left(V,E\right)$ and let $\tilde{L_{\tau}}$
be the Mellin or Laplace transformed $d$-path Laplacian of the graph.
Then, we have the following behavior of the generalized eigenratio

\begin{equation}
Q_{\tau}\rightarrow\left\{ \begin{array}{c}
Q\\
1
\end{array}\right.\begin{array}{l}
\textnormal{if}\ s,~\lambda\rightarrow\infty,\\
\textnormal{if}\ s,~\ensuremath{\lambda\rightarrow}0.
\end{array}
\end{equation}
\end{lemma}
Similarly, for networks of dynamical units with unbounded synchronized region, if one considers synchronizability as measured by the smallest non-zero eigenvalue of the $d$-path Laplacian normalized by the number of nodes, i.e., $\lambda_{2}/N$, the following lemma holds.
\begin{lemma}
\label{lemma2} Let $G=\left(V,E\right)$ and let $\tilde{L_{\tau}}$ be the Mellin or Laplace transformed $d$-path Laplacian of the graph. Then, we have the following behavior of the normalized smallest eigenvalue

\begin{equation}
\left(\dfrac{\lambda_{2}}{N}\right)_{\tau}\rightarrow\left\{ \begin{array}{cc}
\left(\dfrac{\lambda_{2}}{N}\right) & \textnormal{if}\ s,\lambda\rightarrow\infty\\
1 & \textnormal{\textnormal{if}\ s, \ensuremath{\lambda\rightarrow}0}.
\end{array}\right.
\end{equation}
\end{lemma}

In addition, the following Lemma characterizes the behavior of the smallest non-zero eigenvalue and of the largest eigenvalue of the transformed $d$-path Laplacian with respect to the transform parameter ($s$ for the Mellin transform or $\lambda$ for the Laplace transform).

\begin{lemma}
Let $G=\left(V,E\right)$ be a simple graph and let $\tilde{L}_{\tau}\left(G,w\right)$
be the transformed $d$-path Laplacian of $G$ with parameter $w_{ij}$
where $w_{ij}=d_{ij}^{-s}$ ($\tau=$Mellin) and $w_{ij}=e^{-\lambda d_{ij}}$
for $d_{ij}>1$ (\textup{$\tau=$}Laplace). Let $\lambda_{n}\left(G,w\right)$
be the largest eigenvalue of $\tilde{L}_{\tau}\left(G,w\right)$.
Then, if $s'<s$ (Mellin) or $\lambda'<\lambda$ (Laplace) we have
that $\lambda_{2}\left(G,w'\right) \geq \lambda_{2}\left(G,w\right)$ and
$\lambda_{n}\left(G,w'\right)\geq \lambda_{n}\left(G,w\right)$. That
is, \textup{$\lambda_{2}\left(G,w\right)$} and \textup{$\lambda_{n}\left(G,w\right)$}
are non-increasing with $s$ or $\lambda.$
\end{lemma}

\begin{proof}
Let $\vec{x}$ be a nontrivial vector. Using the Rayleigh-Ritz theorem
we have that

$$\lambda_{2}\left(G,w\right)=\underset{\vec{x}\neq0,\vec{x}\bot\vec{1}}\min\dfrac{\vec{x}^{T}\tilde{L}_{\tau}\left(G,w\right)x}{\vec{x}^{T}x}$$
and $$\lambda_{n}\left(G,w\right)=\underset{\vec{x}\neq0}\max\dfrac{\vec{x}^{T}\tilde{L}_{\tau}\left(G,w\right)x}{\vec{x}^{T}x}.$$

Let us select a normalized vector $\vec{x}$ for the sake of simplicity.
Then,
\begin{equation}
\lambda_{2}\left(G,w\right)=\underset{\vec{x}\neq0,\vec{x}\bot\vec{1}}{\min}\sum_{i,j}w_{ij}\left(x_{i}-x_{j}\right)^{2},
\end{equation}

\begin{equation}
\lambda_{n}\left(G,w\right)=\underset{\vec{x}\neq0}{\max}\sum_{i,j}w_{ij}\left(x_{i}-x_{j}\right)^{2}.
\end{equation}
Then, for a fixed vector $\vec{x}$, a decrease of the parameter $w_{ij},$
i.e., increase of the parameters $s$ or $\lambda$, makes that $\lambda_{2}\left(G,s\right)$
and $\lambda_{n}\left(G,s\right)$ decay or remain at their previous value, so in general they are non-increasing with $s$ or $\lambda$.
\end{proof}

Based on these considerations and on the fact that $\lambda_2\geq \lambda_n$, one could expect that $Q_\tau$ is also increasing on average when $s$ or $\lambda$ are decreasing. Numerical results, which are shown in the next section, along with the fitting functions, show that, in addition, the decay is monotonic. This implies that for any network
of coupled oscillators, independently of its topology and the unit
dynamics, the increase of the long-range coupling between oscillators,
i.e., $s,\lambda\rightarrow0$, produces the best possible synchronizability.

We illustrate this result by calculating $\lambda_{2,\tau}/N$ and $Q_{\tau}$
for Erd\H{o}s-Renyi (ER) random graphs with $N=100$ and $\langle k\rangle=4$
for increasing values of the parameters $s$ and $\lambda$ of the
$d$-path Laplacian matrices. Fig.~\ref{fig:syncbltER} clearly shows
that for $s,\lambda\rightarrow0$ the best possible synchronizability
is obtained, whereas for $s,\lambda\rightarrow\infty$ the synchronizability
of the original graph is recovered.

\begin{figure}
\centering{}
\subfigure[]{\includegraphics[width=0.4\textwidth]{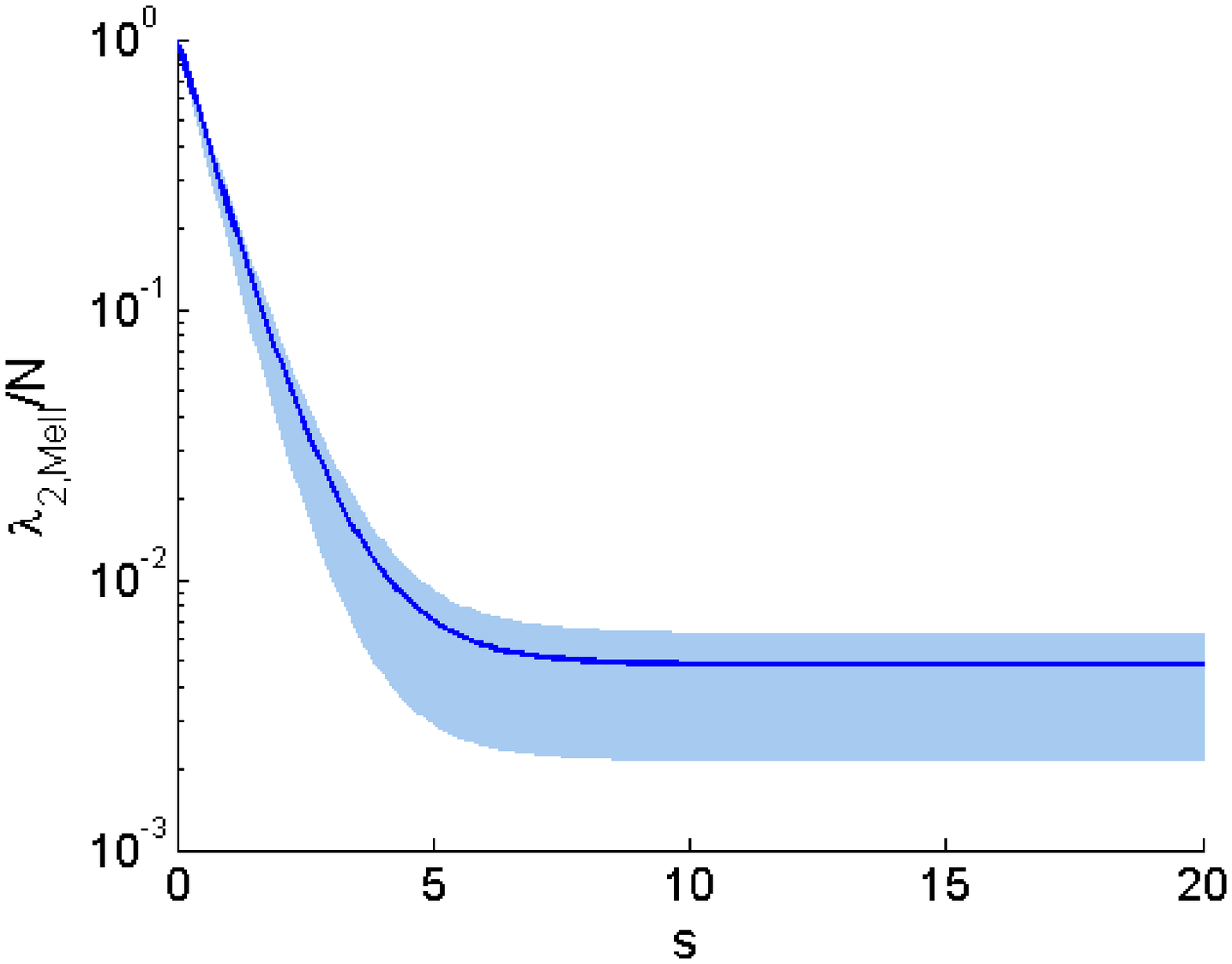}}
\subfigure[]{\includegraphics[width=0.4\textwidth]{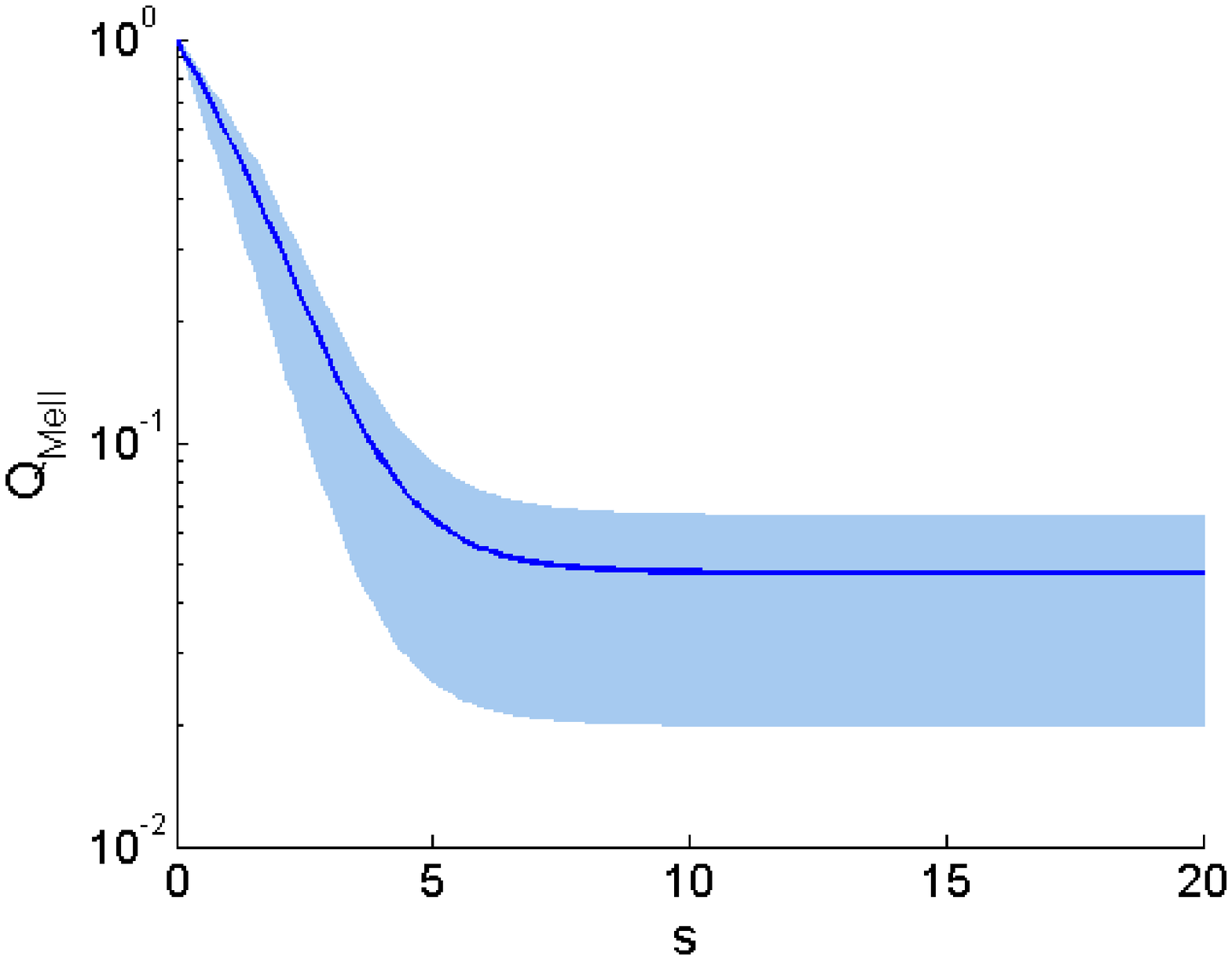}}\\
\subfigure[]{\includegraphics[width=0.4\textwidth]{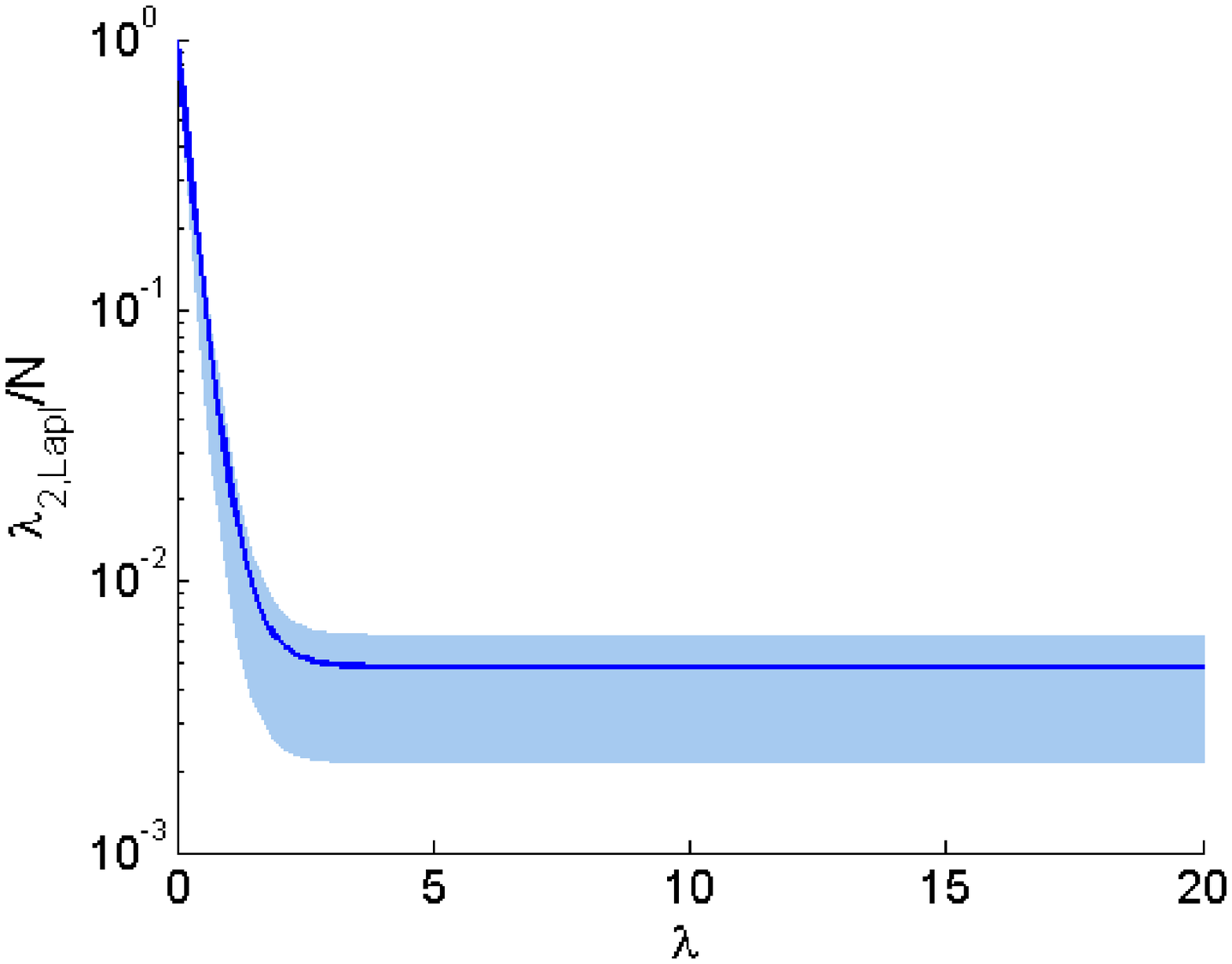}}
\subfigure[]{\includegraphics[width=0.4\textwidth]{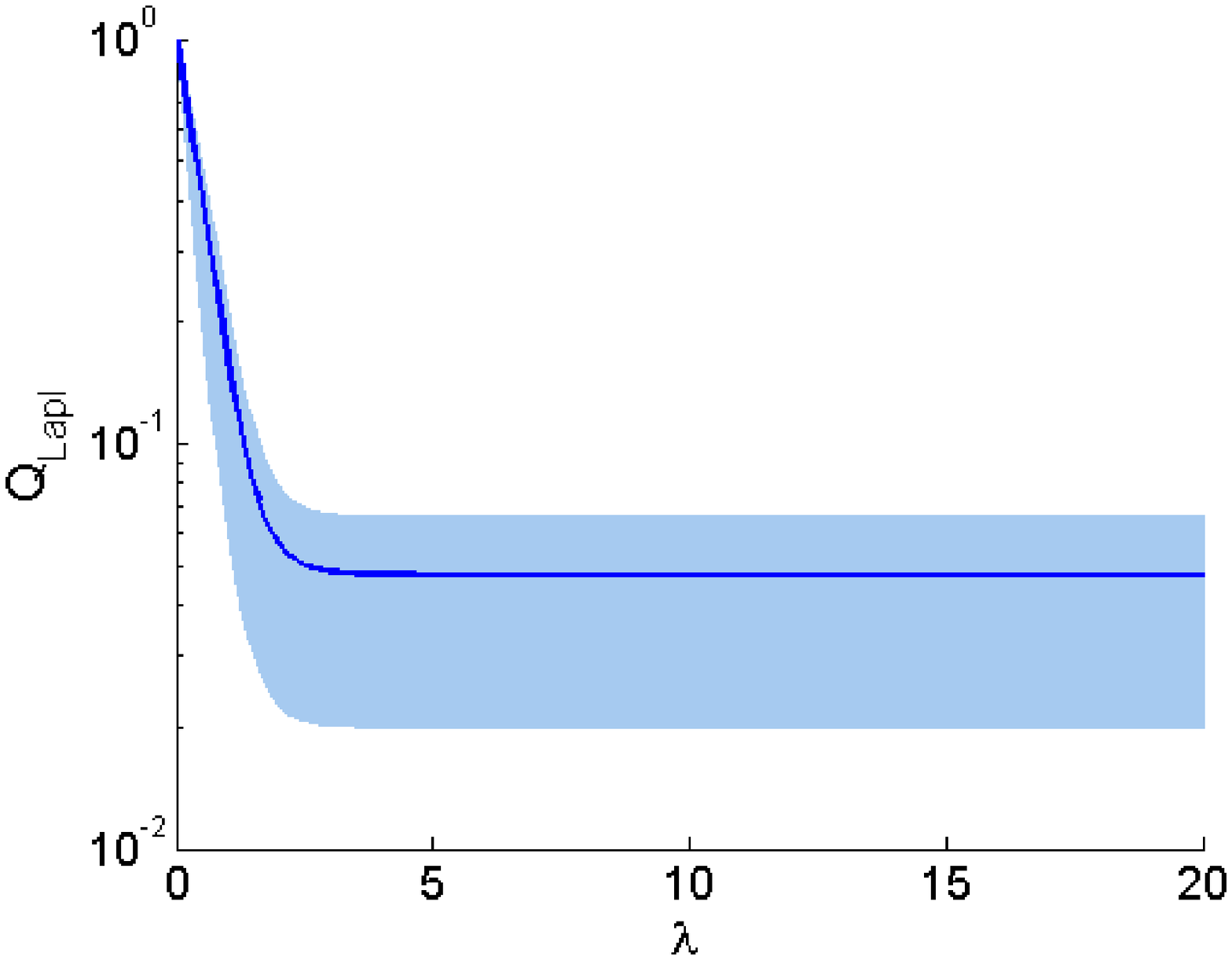}}
\caption{\label{fig:syncbltER} Synchronizability of LRI ER networks. The $d$-path
Laplacian matrices are Mellin transformed in (a)-(b) and Laplace transformed
in (c)-(d). Panels (a) and (c) illustrates $\lambda_{2}/N$, whereas
panels (b) and (d) illustrate $\lambda_{2}/\lambda_{N}$. Results
are averaged on 50 networks with $N=100$ and $\langle k\rangle=4$. The solid lines represent mean value of the measures for synchronizability, while the shadow regions indicate the range of variability (minimum and maximum values).}
\end{figure}

A note of caution should be written here. The fact that for $s,\lambda\rightarrow0$ the network behaves like a complete graph does not mean that the current results are trivial in the sense that we replace the graph under study by a complete graph. If we consider values of $s,\lambda$ close but not exactly equal to zero, the synchronizability of the networks approach the best possible value, but (and this is an important but) it still depends on the topology of the network. In order to illustrate this
fact we investigate next how the values of $\lambda_{2,\tau}/N$ and $Q_{\tau}$ change with $s$ or $\lambda$ for a few different types of graphs. In particular, we have done the calculations for a series of graphs: the star topology, the triangular and square lattices, the ring, the Barbell graph and the path graph. For each of these graphs we have applied the Mellin and the Laplace transform and computed the synchronizability measures for different values of $s$ or $\lambda$. The results, shown in Fig. \ref{fig:regulargraphs}, include as reference the synchronizability of the complete graph, which clearly does not depend on the transform parameters. The curves clearly show that the synchronizability tends to one for all the graphs as LRIs are increasingly weighted ($s,\lambda \rightarrow 0$), but the impact of LRIs is a function of the graph topology.

\begin{figure}
\centering{}
\subfigure[]{\includegraphics[width=0.4\textwidth]{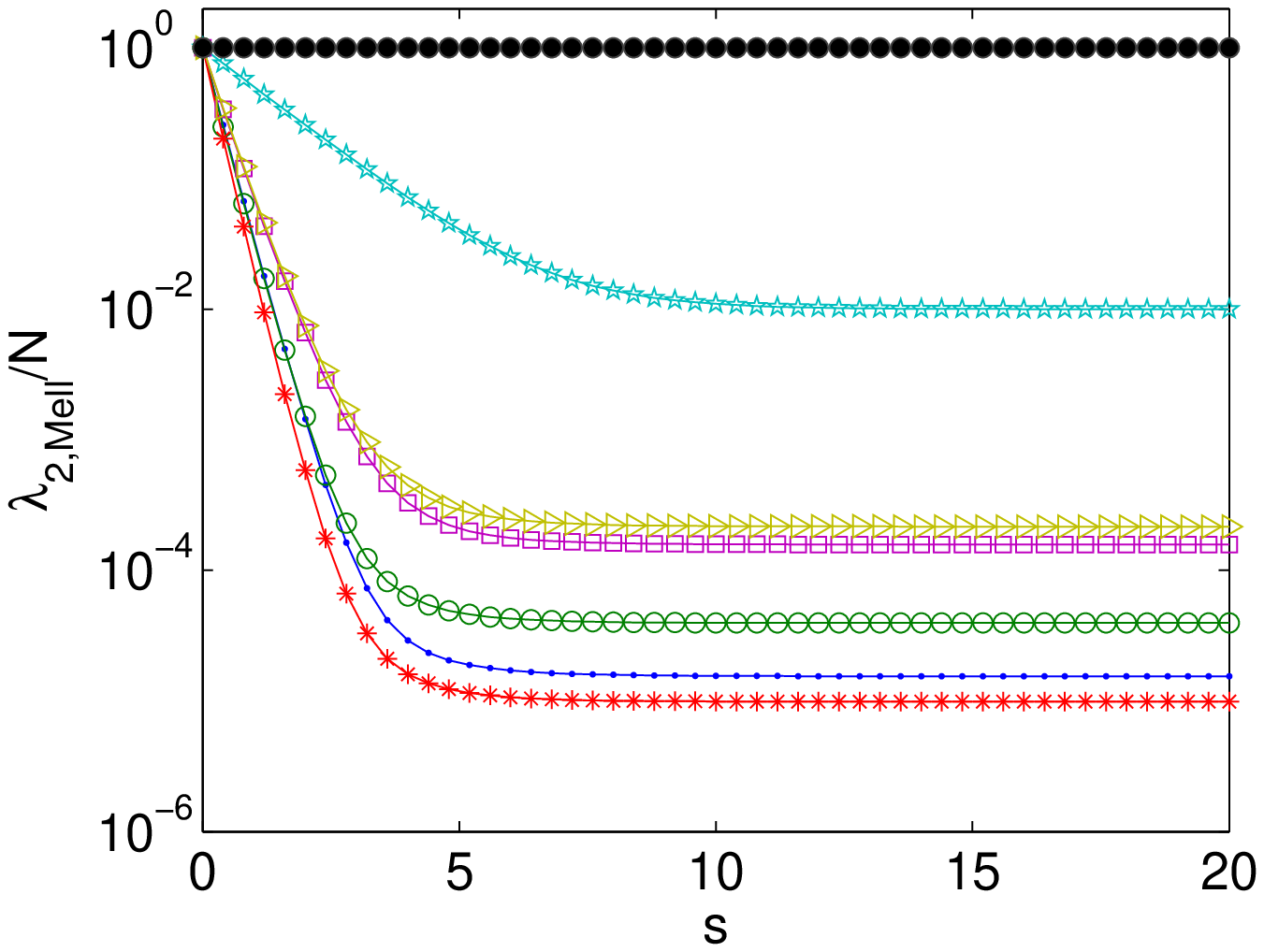}}
\subfigure[]{\includegraphics[width=0.4\textwidth]{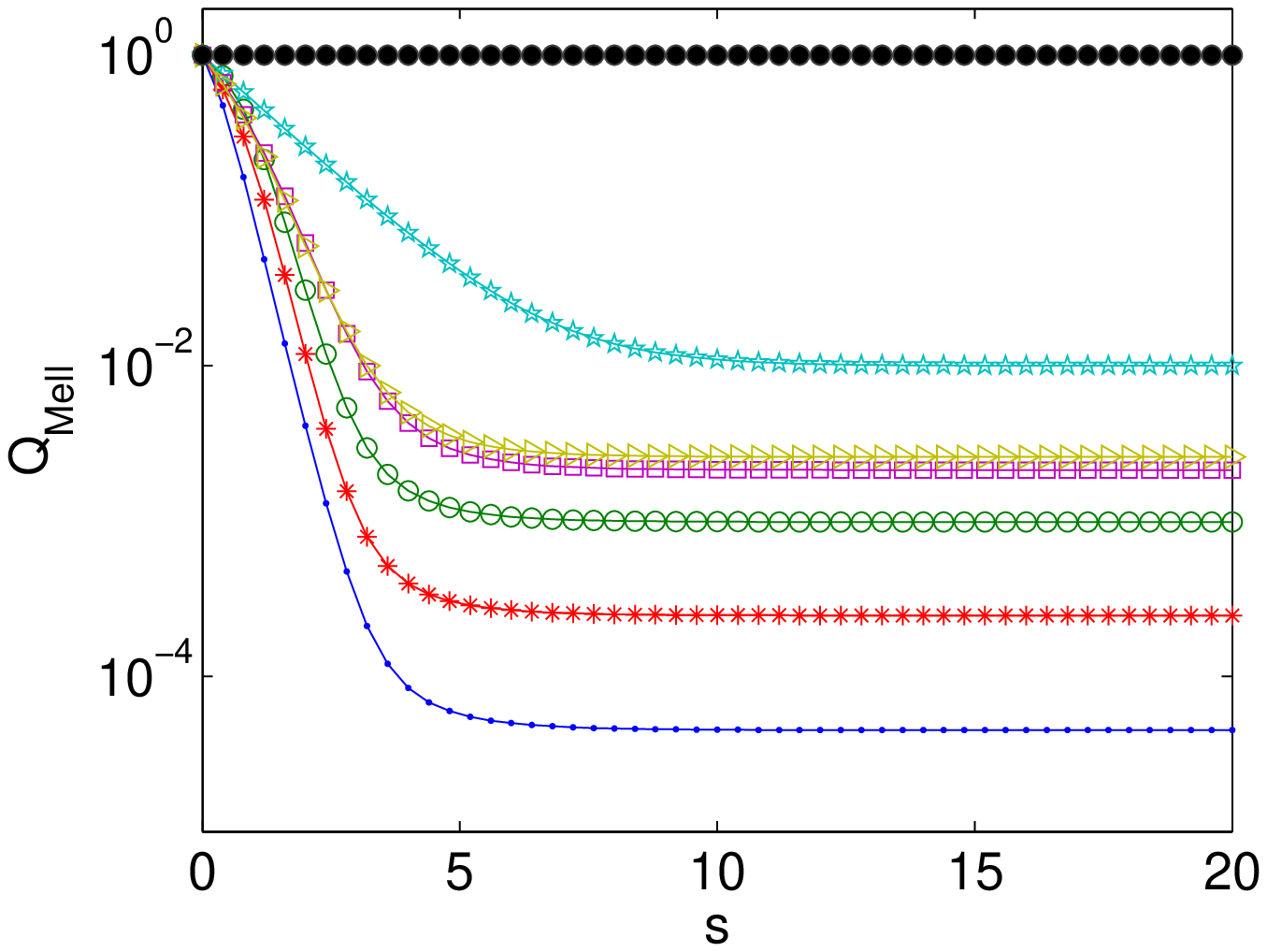}}\\
\subfigure[]{\includegraphics[width=0.4\textwidth]{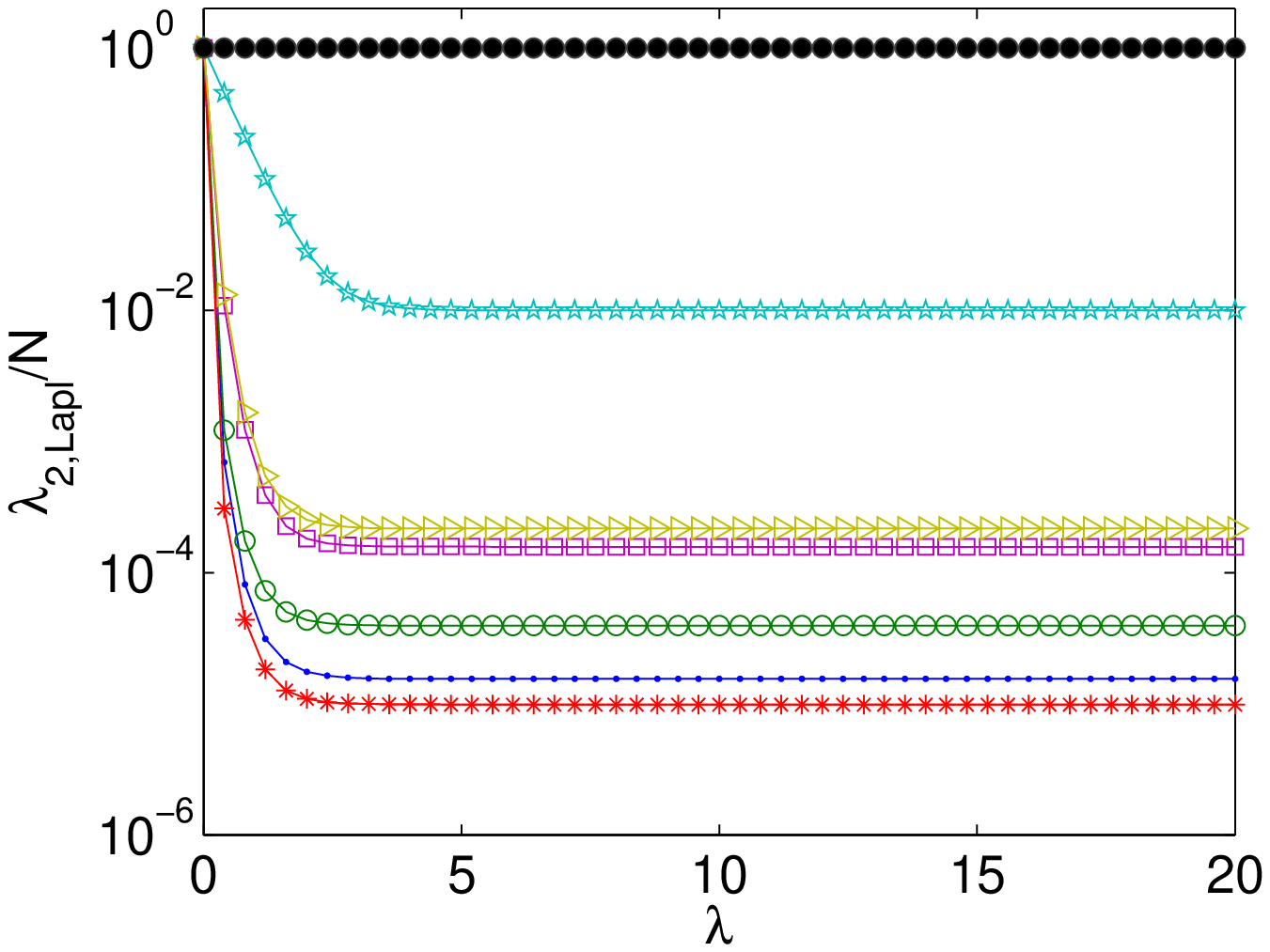}}
\subfigure[]{\includegraphics[width=0.4\textwidth]{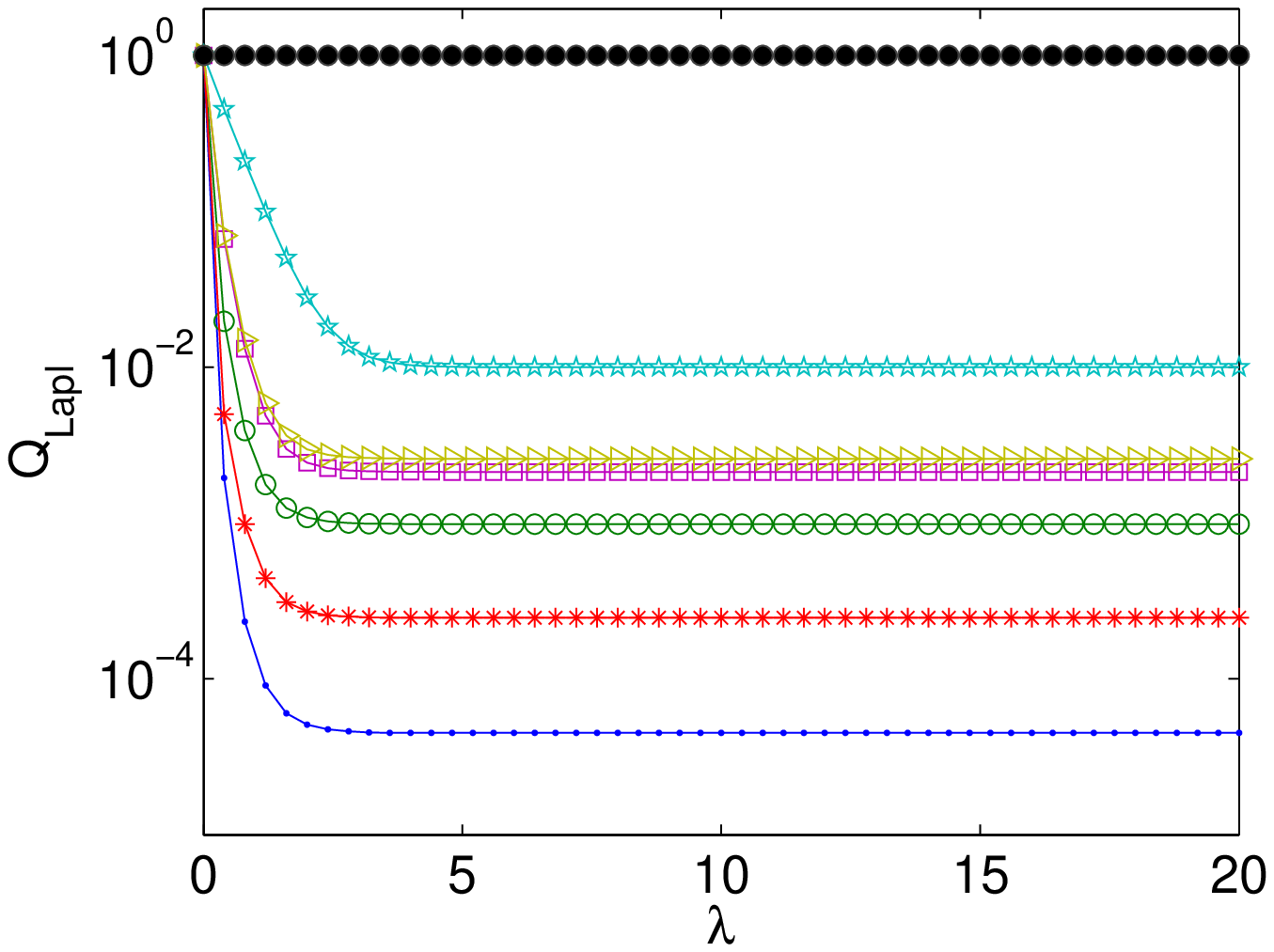}}
\caption{\label{fig:regulargraphs} Synchronizability of a few graphs with LRIs. The $d$-path
Laplacian matrices are Mellin transformed in (a)-(b) and Laplace transformed
in (c)-(d). Panels (a) and (c) illustrates $\lambda_{2,\tau}/N$, whereas
panels (b) and (d) illustrate $Q_{\tau}$. The different curves correspond to these graphs: complete graph (black, filled circles); star topology (cyan, stars); triangular lattice (yellow, triangles); square lattice (magenta, squares); ring (green, open circles); Barbell graph (blue, dots); path graph (red, asterisks).}
\end{figure}

To further corroborate our finding, we considered a dataset of 54
real-world networks (see Appendix for descriptions of the networks)
and calculated $\lambda_{2,\tau}/N$ and $Q_{\tau}$ without and with LRIs. Fig.~\ref{real_world_global}(a) shows the results of the change of $Q_{\tau}$ \textemdash the results
for $\lambda_{2,\tau}/N$ are very similar\textemdash with the parameter
$s$ of the Mellin transform. This confirms our previous result
that the rate of change of the spectral ratio with the parameters
of the transforms used in the synchronization models depends on the
topology of the network. Just to give a clear example of the meaning
of this finding, when $s=1$ there are networks (such as SoftwareMysql, network 17 of Table~\ref{tab:tab1} in Appendix) for which the relative value of $Q_{\tau}$ is approximately
30\% of the best possible value while for others (such as Skipwith, network 29 of Table~\ref{tab:tab1} in Appendix) it approaches 60\%. This demonstrates that the topology of the network
continues to play a fundamental role in the synchronizability of the
system even when the long-range influence of the nodes is very high.
Knowing how the network topology influences the rate of change of
the spectral parameters and consequently of their synchronizability
is an important open question that should be further investigated. Some insights on this is provided by the following considerations on how to fit the data of Fig.~\ref{real_world_global}(a). To illustrate this, let us first consider the star topology. For this network we have been able to find a closed formula for the eigenvalues of the $d$-path transformed Laplacian matrix. Let $\lambda_i(\mathcal{L})$, $\lambda_i(\mathcal{L}_d)$, and $\lambda_i(\mathcal{K}_n)$ indicate the $i$-th eigenvalue of the original Laplacian matrix, of the $d$-path transformed Laplacian matrix (here, for simplicity we consider only the Mellin transform, but similar results hold for the Laplace one) and of the Laplacian matrix of the complete graph. For the star topology, we have:

\begin{equation}
\lambda_i(\mathcal{L}_d)=\lambda_i(\mathcal{K}_n)p^{-s}+\lambda_i(\mathcal{L})(1-p^{-s})
\end{equation}

\noindent with $p=2$ (this expression has been analytically checked up to $n=6$ and numerically verified for larger $n$).

We have then generalized this formula for an arbitrary network as follows:

\begin{equation}
\lambda_i(\mathcal{L}_d)=\lambda_i(\mathcal{K}_n)p_i^{-s}+\lambda_i(\mathcal{L})(1-p_i^{-s})
\end{equation}

\noindent where now $p_i$ is a fitting parameter (one for each eigenvalue to fit) that depends on the topology we are investigating. Note that for $s=0$ we have $\lambda_i(\mathcal{L}_d)=\lambda_i(\mathcal{K}_n)=n$ and for $s\rightarrow \infty$ we have $\lambda_i(\mathcal{L}_d)=\lambda_i(\mathcal{L})$ (the eigenvalues of the original network).

Finally we propose to fit $Q_{Mell}$ as:

\begin{equation}
\label{eq:fitting}
Q_{Mell}=\frac{\lambda_2(\mathcal{K}_n)p_2^{-s}+\lambda_2(\mathcal{L})(1-p_2^{-s})}{\lambda_n(\mathcal{K}_n)p_n^{-s}+\lambda_n(\mathcal{L})(1-p_n^{-s})}
\end{equation}

An example for three real-world networks is shown in Fig.~\ref{real_world_global}(b). This fitting provides an approximate formula describing how $Q_{Mell}$ changes with $s$. Noticeably, the formula depends on the characteristics of the topology through the eigenvalues $\lambda_2$ and $\lambda_n$ of the original network Laplacian and some fitting parameters $p_2$ and $p_n$.

\begin{figure}
\begin{centering}
\subfigure{\includegraphics[width=0.4\textwidth]{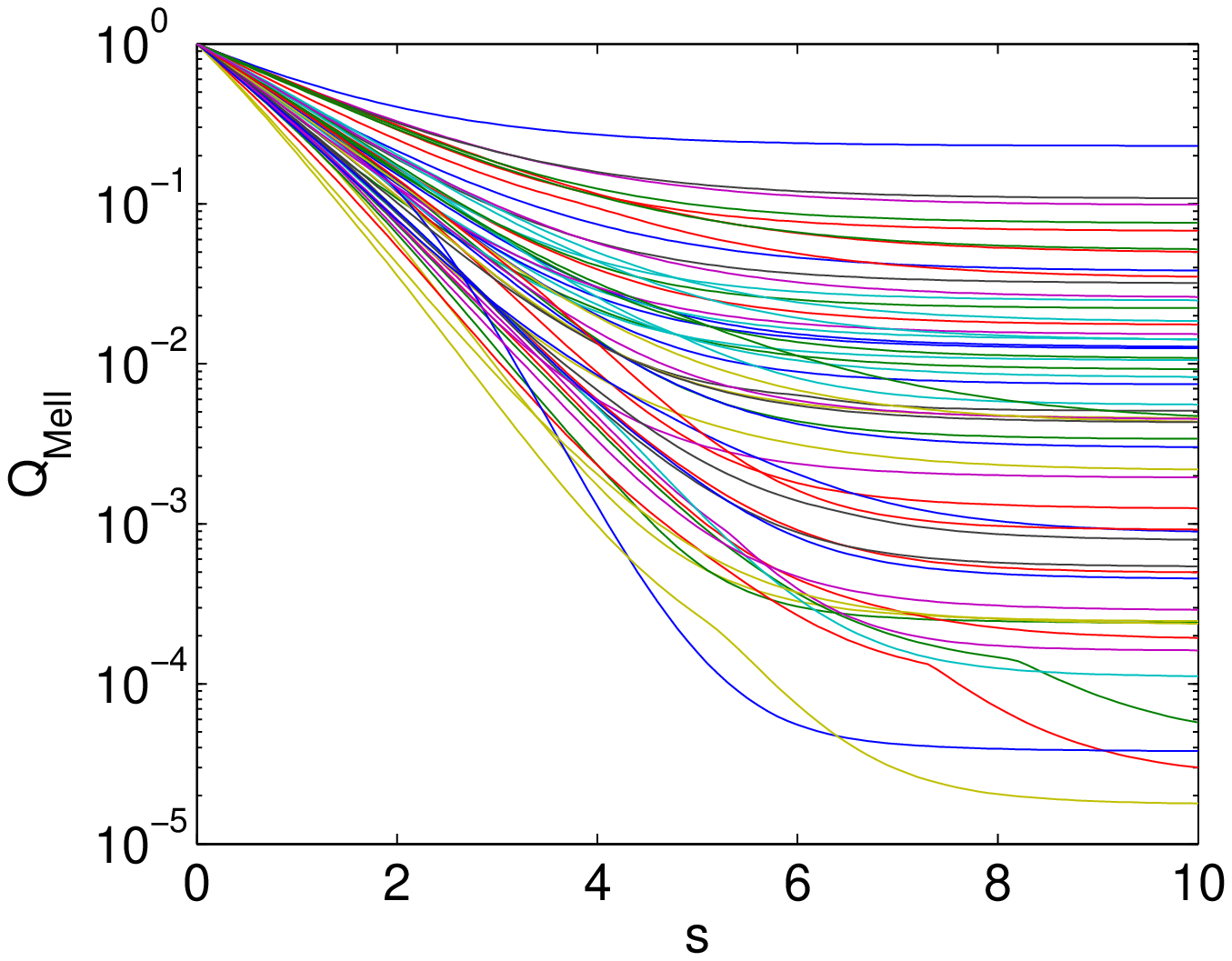}}
\subfigure{\includegraphics[width=0.4\textwidth]{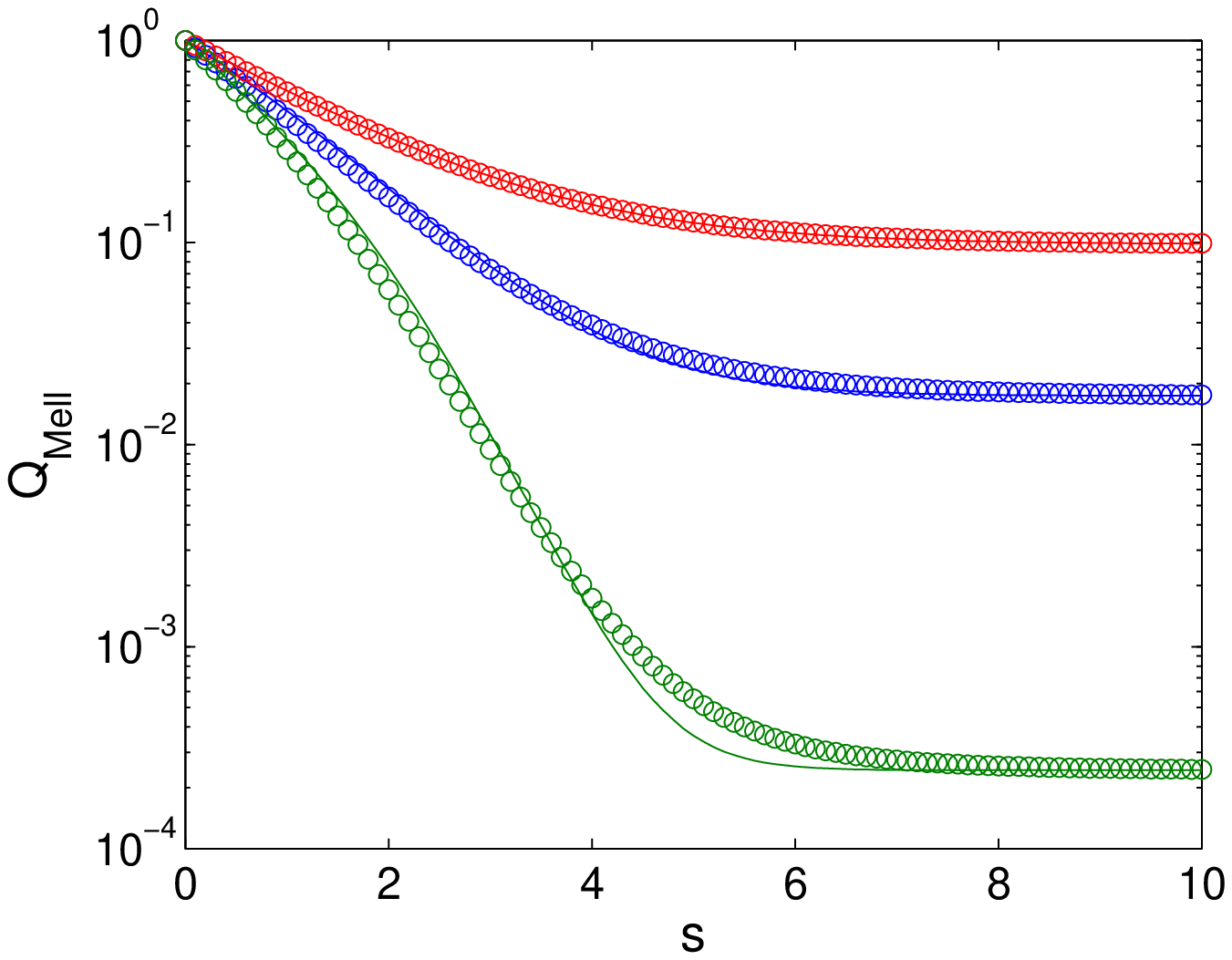}}
\par\end{centering}
\caption{(a) Change of the synchronizability parameter (spectral ratio) for graphs with LRIs as a function of the Mellin transform parameter $s$ for 54 real-world networks studied here. (b) Fitting of the synchronizability parameter (spectral ratio) for three real-world graphs (network 33 in blue, network 49 in red, network 51 in green). Numerical data are reported with symbols, while the continuous line is the fitting through Eq. (\ref{eq:fitting}) with $p_2=3.16$ and $p_n=2.08$ for network 33, $p_2=2.07$ and $p_n=2.00$ for network 49, and $p_2=10.45$ and $p_n=4.05$ for network 51.}
\label{real_world_global}
\end{figure}

In the case of the Kuramoto model we show some further results confirming the beneficial effect of LRIs. For this model, in fact, it is known that $\lambda_{2}$ plays a fundamental role. In particular, for identical oscillators, it is known that the synchronization time scales with the inverse of $\lambda_{2}$. Based on this and on the result of Lemma~\ref{lemma2}, we thus expect that LRIs in the Kuramoto model promote synchronization. For the numerical simulations, we considered two different network topologies, ER and scale-free (SF) random graphs \cite{albert2000error}, and monitored synchronization by measuring the order parameter $r$, defined as

\begin{equation}
r=\left\langle\left|\frac{1}{N}\sum\limits _{i=1}^{N}e^{j\theta_{i}}\right|\right\rangle_{T},
\end{equation}

\noindent where $T$ is a sufficiently large averaging window. Values
of $r$ close to one indicates a high degree of synchronization, while
low values (close to zero) the absence of coherence among the oscillators.
Fig.~\ref{fig:Kuramoto}(a) shows the behavior of the order parameter
$r$ vs. the coupling $\sigma$ for an ER network with $N=100$ and
different values of $s$. As expected, decreasing $s$ favors synchronization.
Similar results are obtained for SF networks (Fig.~\ref{fig:Kuramoto}(b)).
The results are illustrated for Mellin transformed $d$-path Laplacian
matrices, but also hold when the Laplace transform is considered.

\begin{figure}
\centering{}
\subfigure[]{\includegraphics[width=0.4\textwidth]{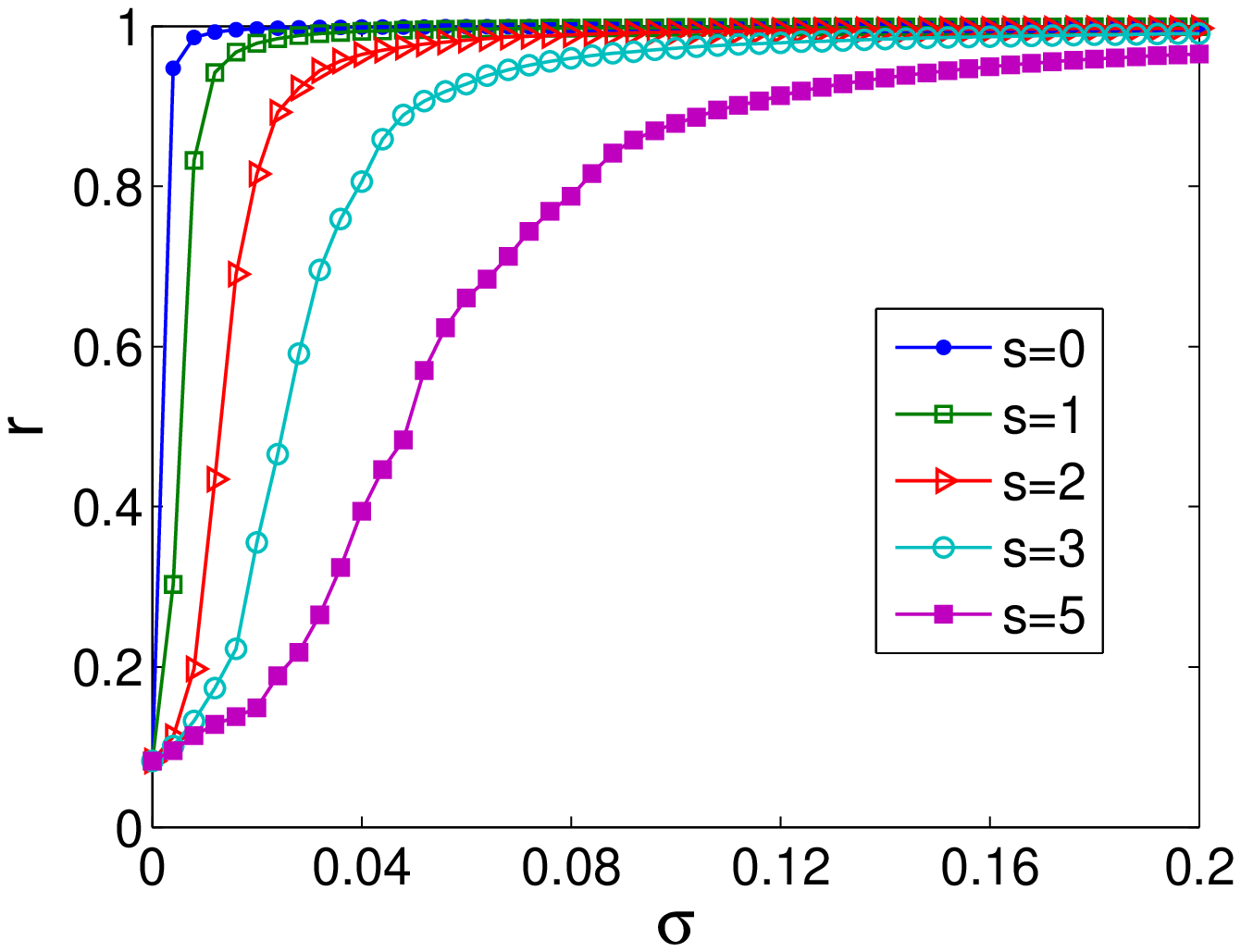}}
\subfigure[]{\includegraphics[width=0.4\textwidth]{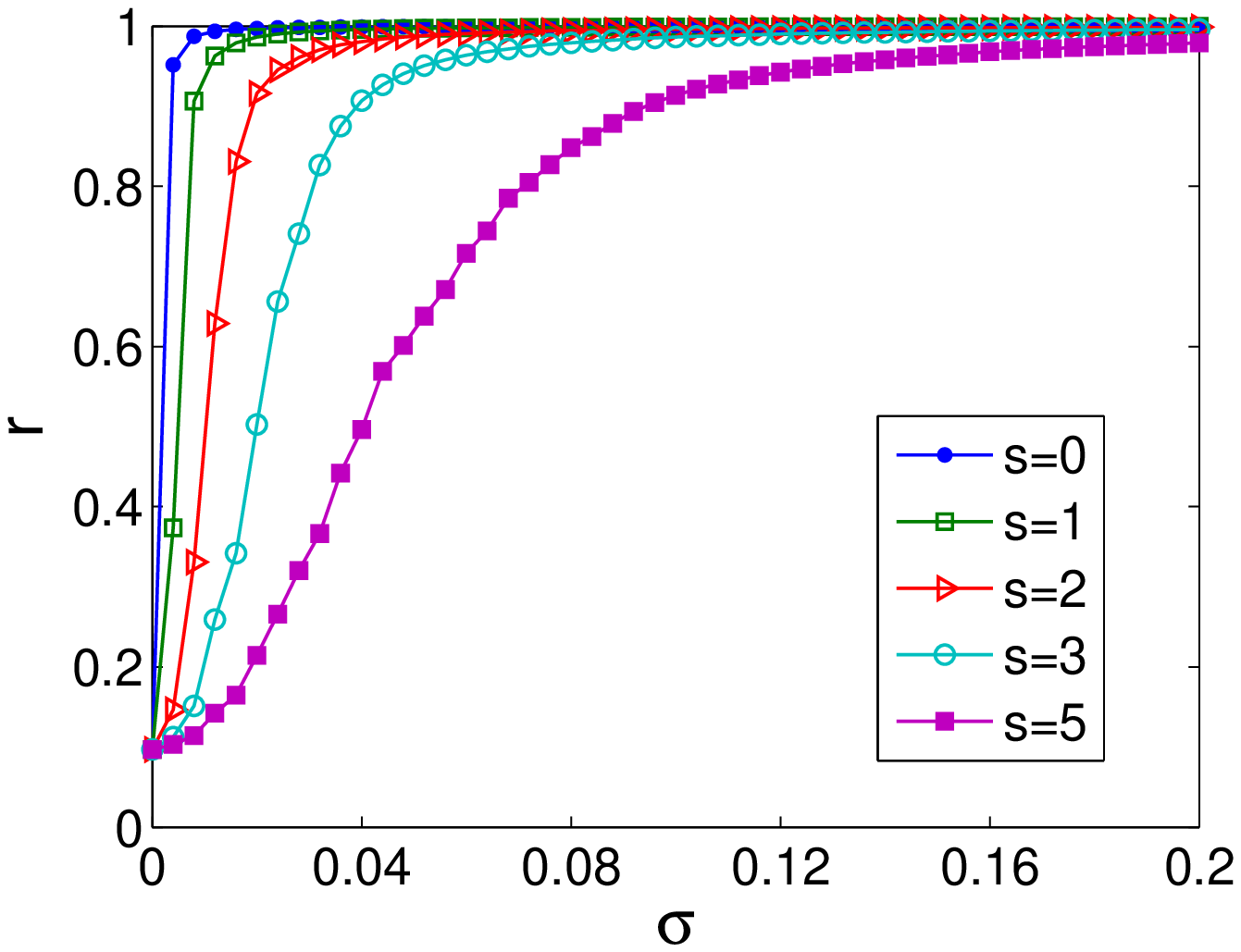}}
\caption{\label{fig:Kuramoto} Synchronization order parameter $r$ vs. the
coupling strength $\sigma$ for an ER (a) and SF (b) network of $N=100$ Kuramoto oscillators with LRIs. Mellin transform is used to weight the LRIs.}
\end{figure}

\section{How does network topology influence the role of LRIs?}

The analysis of the effects of LRIs on the synchronizability of real-world networks has revealed that the quantitative impact of LRIs varies from network to network. In this section we study how the effects of LRIs on synchronizability depend on some network characteristics by considering a series of artificial networks with controlled attributes.

To begin, we start investigating how the effects of LRIs depend on network diameter. To this aim, we have considered the Watts-Strogatz model generating small-world networks through a rewiring process applied to a pristine regular graph \cite{watts1998collective}. More specifically, we started from a ring with 4-nearest neighbors and then rewired the links with a progressively increasing rewiring probability, labeled as $p_{r}$. This produces networks with the same number of nodes and links, but with a different
diameter (large for $p_{r}=0$, small for $p_{r}=1$). Figs.~\ref{fig:alphanorm}(a)
and (d) illustrate the results. LRIs always lead to enhancement of
synchronizability, but the larger is the diameter the larger are the
(beneficial) effects of LRIs.

We have then studied how the effects of LRIs are influenced by the
average degree $\langle k\rangle$. To this aim, we have considered
networks with the same number of nodes and a growing number of links.
The ER model is adopted. Figs.~\ref{fig:alphanorm}(b) and~(e) shows
that the smaller is the average degree the larger are the effects
of LRIs on synchronizability, so that we conclude that LRIs are more
important in networks with smaller degree.

Finally, we have investigated the effect of degree heterogeneity by
simulating networks with the same number of nodes, number of links
and diameter, but different degree distributions. These are obtained
by using the network model described in \cite{gomez2006scale}, which
parameterizes with $\alpha$ the tuning from one network type to the
other ($\alpha=0$ corresponds to a SF network, while $\alpha=1$
to an ER network). Fig.~\ref{fig:alphanorm}(c) and~(f) shows that
the effects on $\lambda_{2,\tau}/N$ poorly depend on the topology, but those
on the ratio $Q_{\tau}$ are much more important on
SF networks rather than in ER structures. So, SF networks receive
more benefits from the inclusion of LRIs.

\begin{figure}
\centering{}
\subfigure[]{\includegraphics[width=0.32\textwidth]{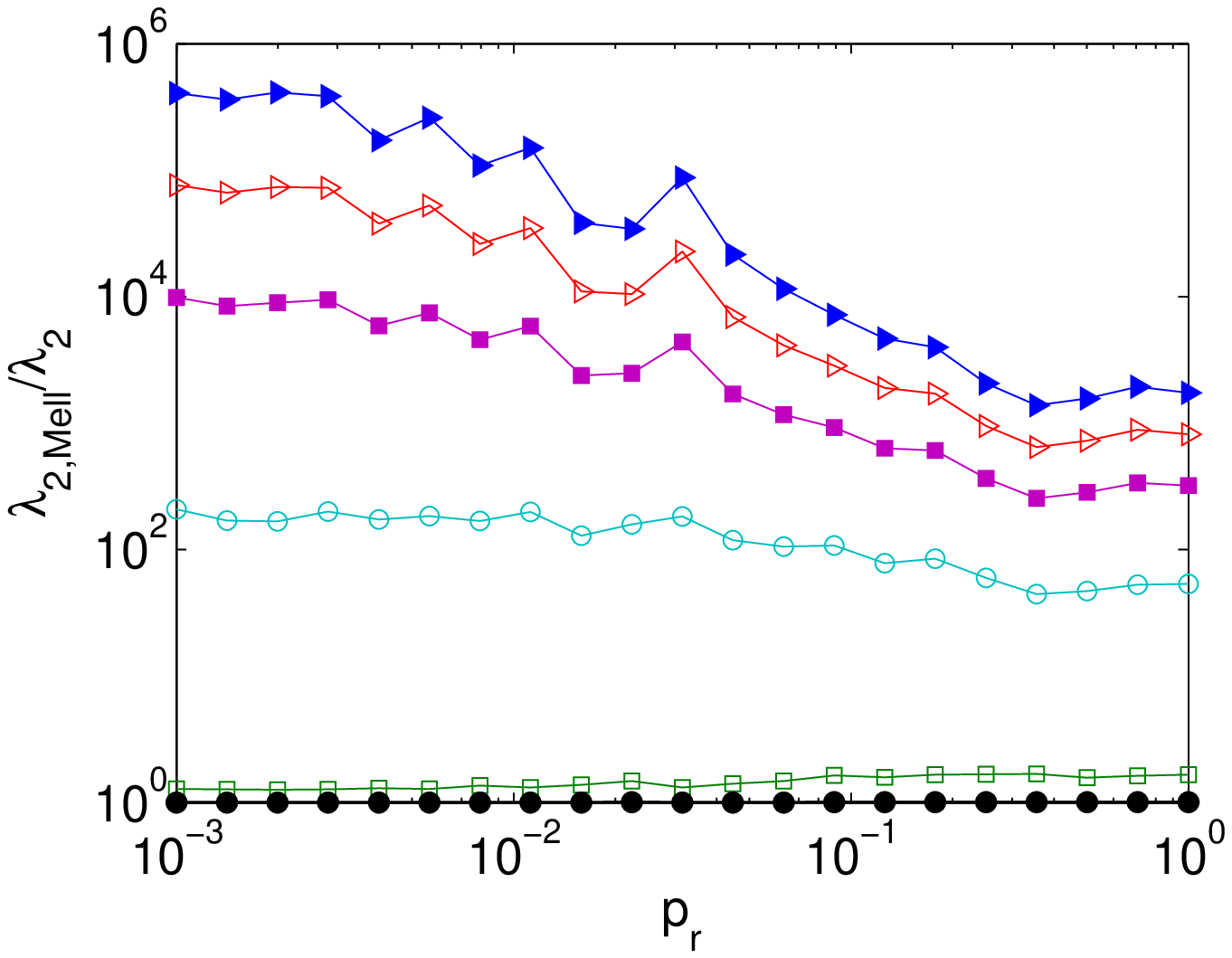}}
\subfigure[]{\includegraphics[width=0.32\textwidth]{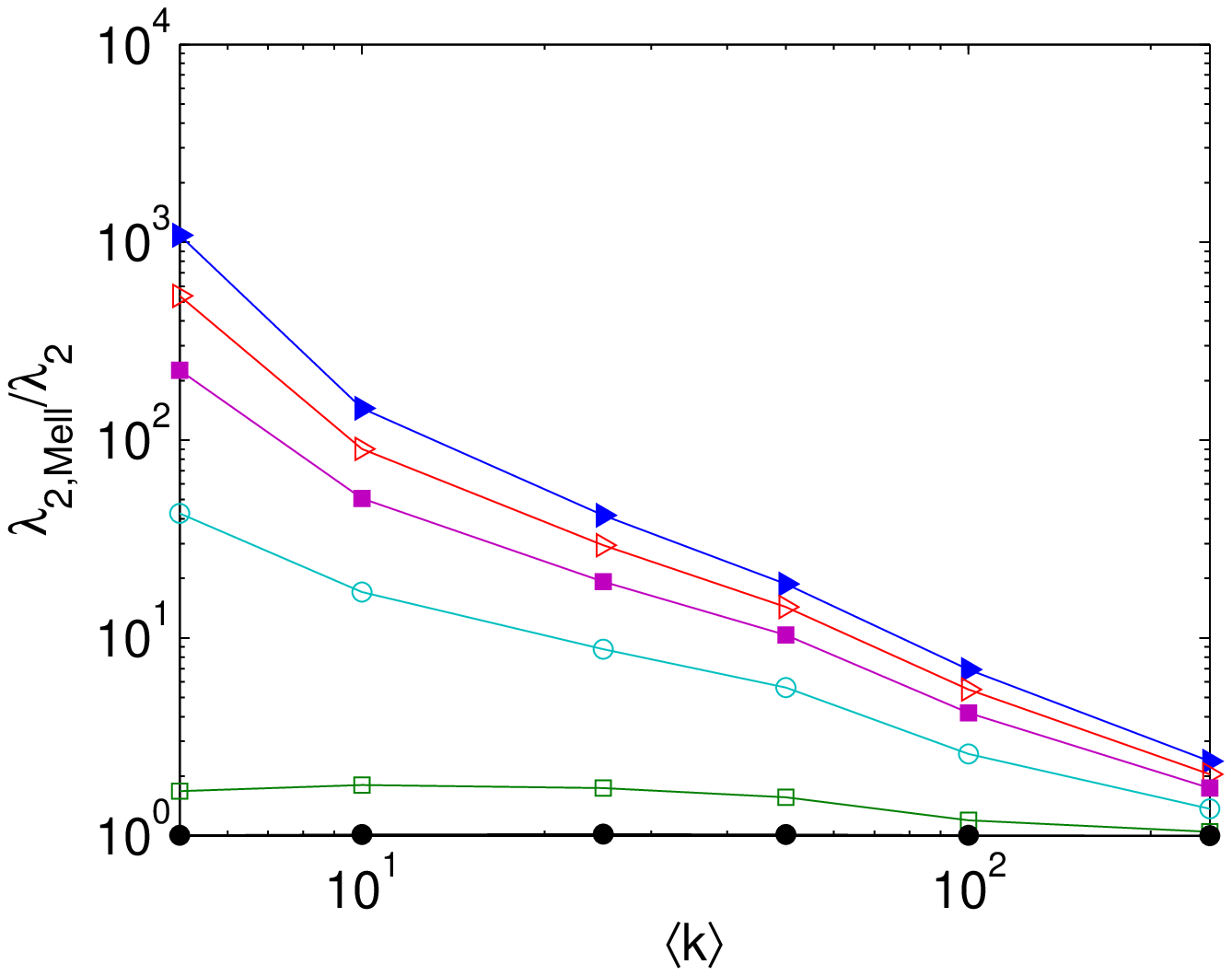}}
\subfigure[]{\includegraphics[width=0.32\textwidth]{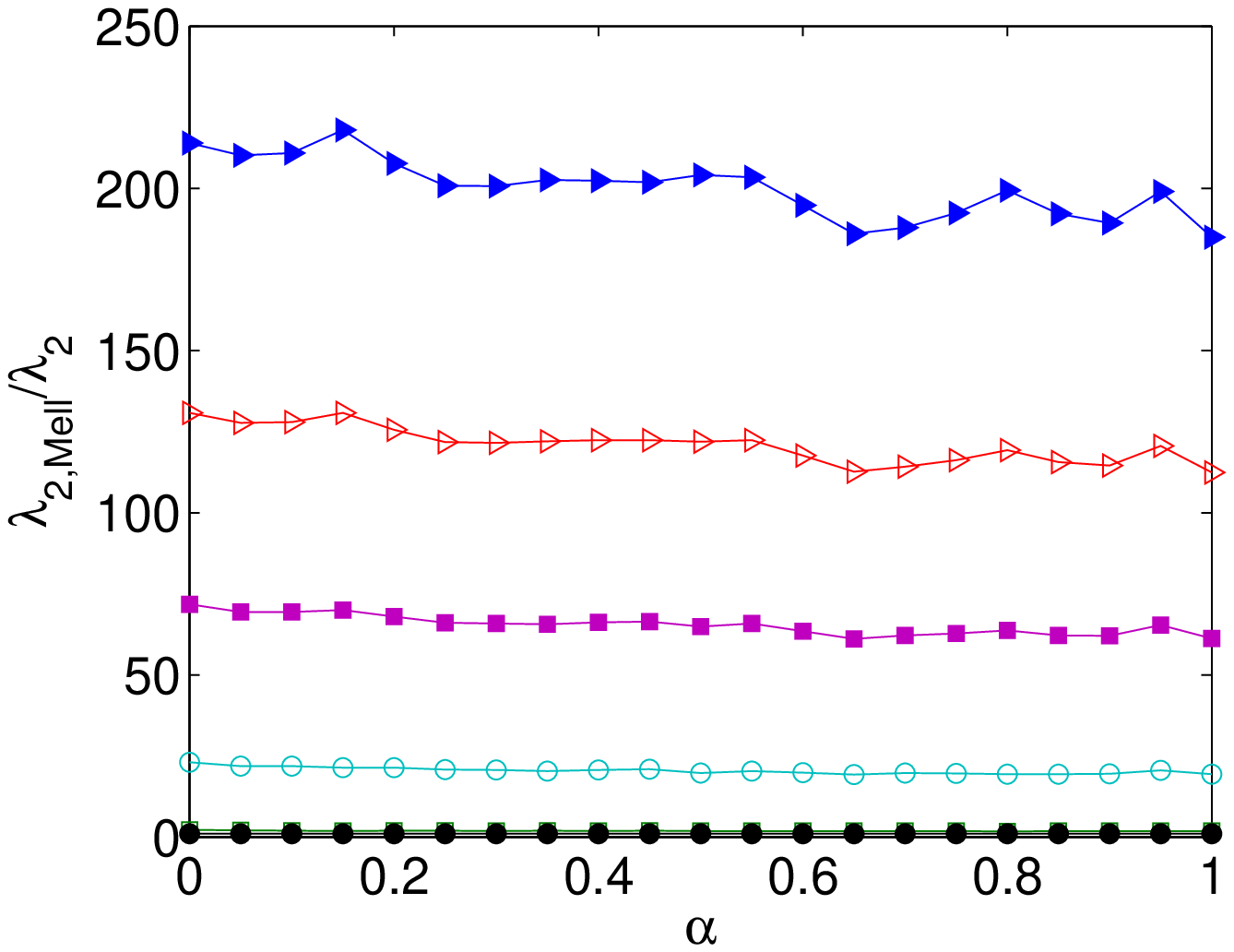}}
\subfigure[]{\includegraphics[width=0.32\textwidth]{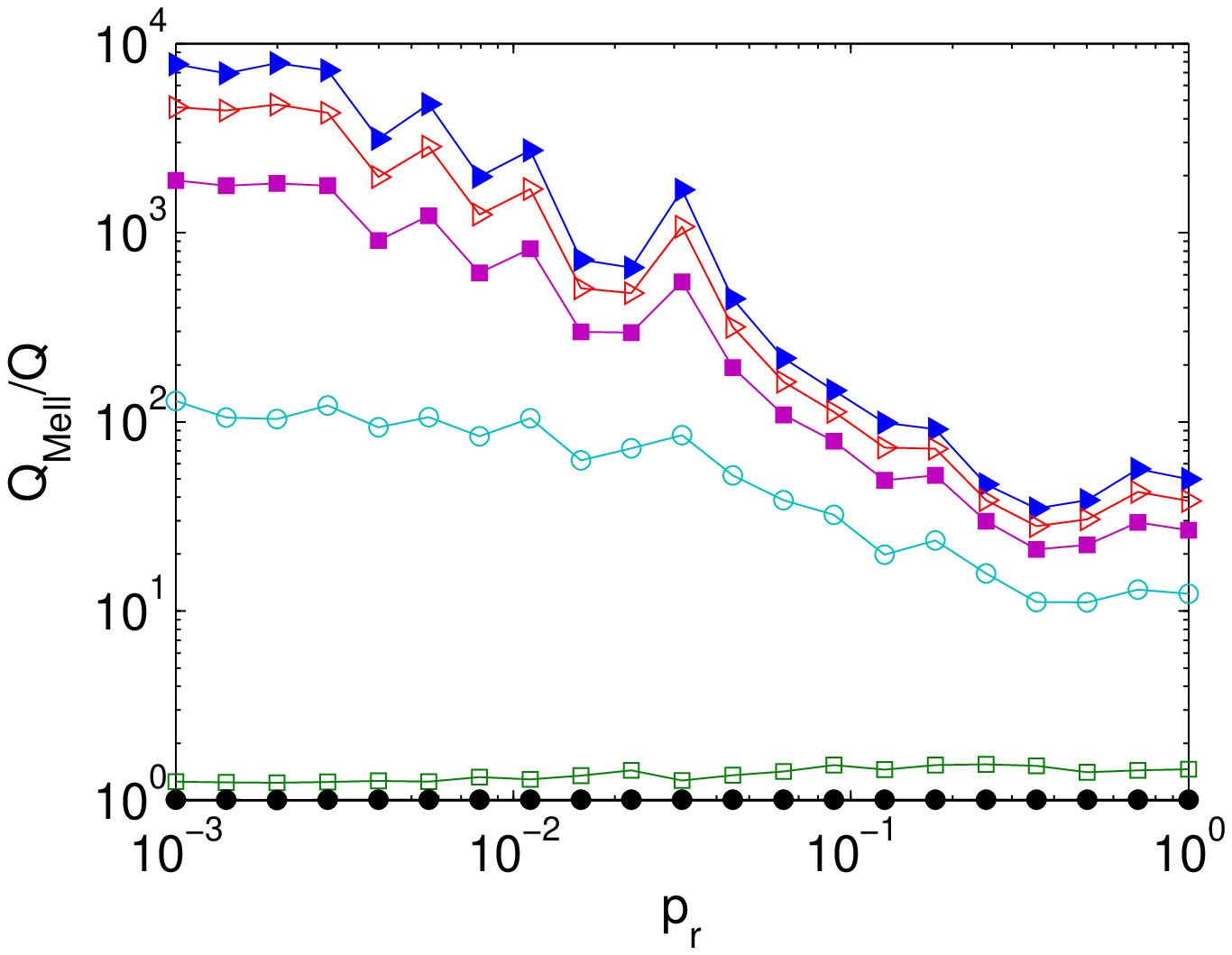}}
\subfigure[]{\includegraphics[width=0.32\textwidth]{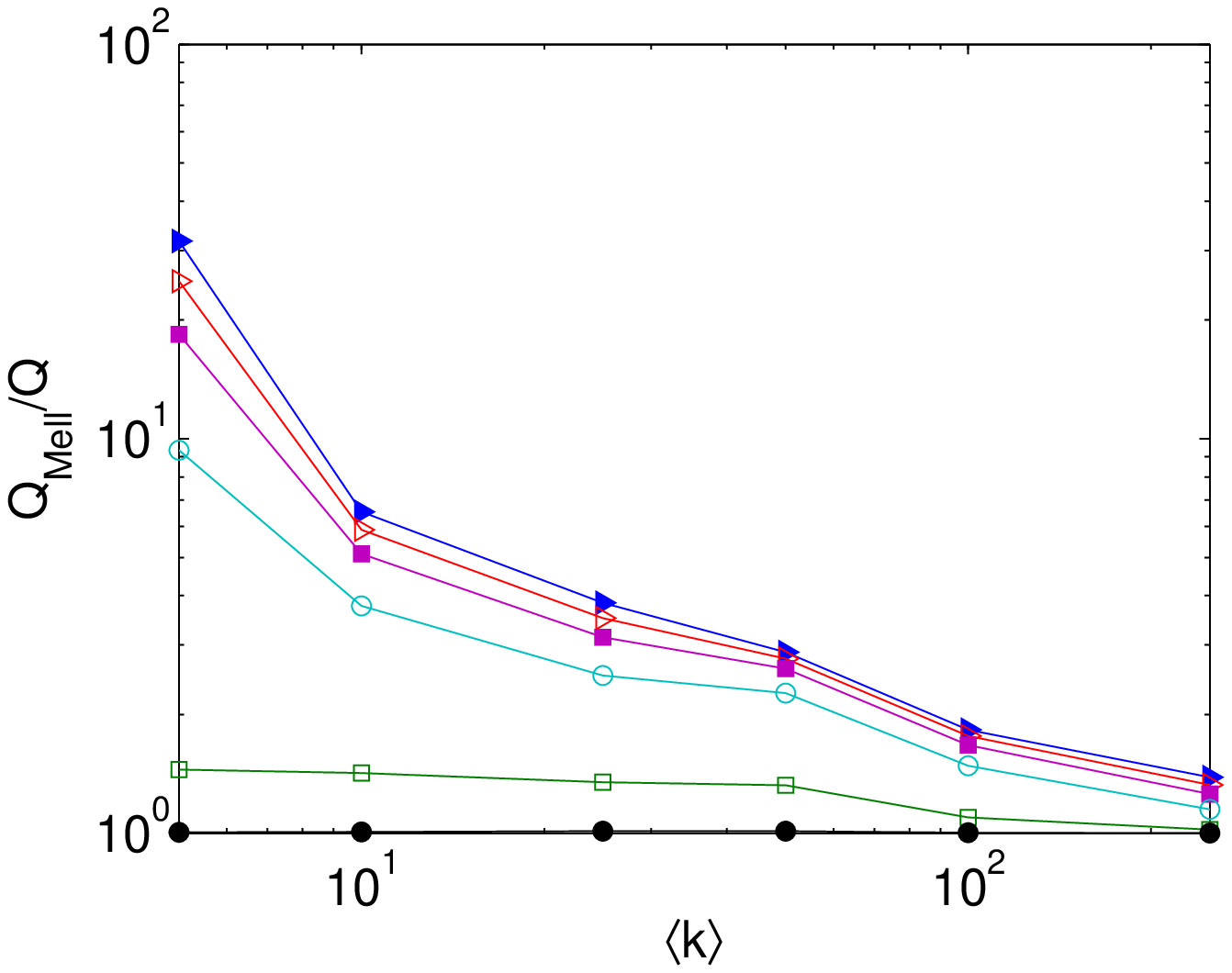}}
\subfigure[]{\includegraphics[width=0.32\textwidth]{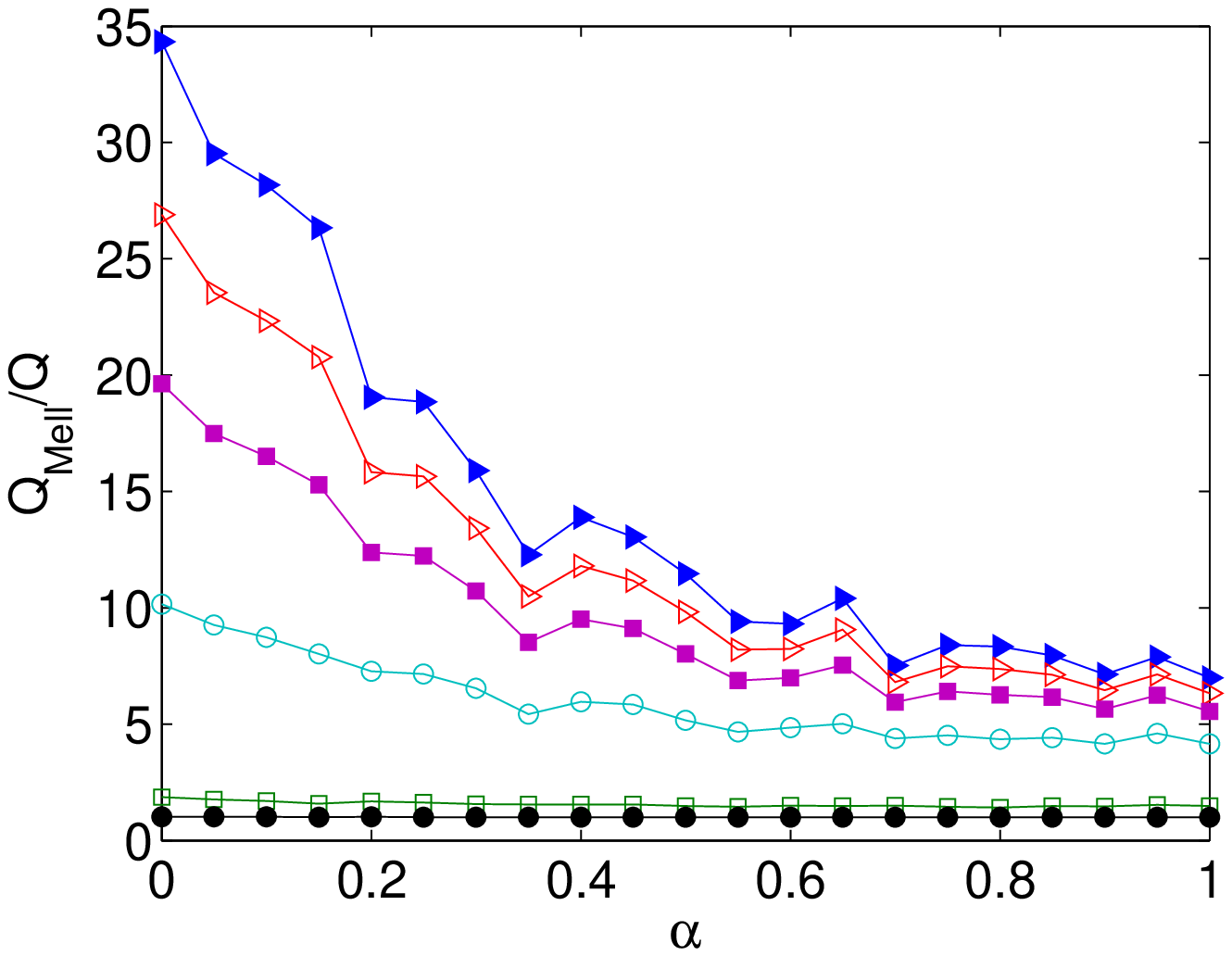}}
\caption{\label{fig:alphanorm} Values of $\lambda_{2,Mell}/\lambda_{2}$ (a)-(c) and of $Q_{Mell}/Q$ (d)-(f) for artificial networks with LRIs. LRIs are weighted with the Mellin transform. Panels (a) and (d) refer to networks with different diameter (see text for the details on the generation of the networks). Panels (b) and (e) refer to artificial networks with different average degree $\langle k\rangle$. Panels (c) and (f) refer to artificial networks with different heterogeneity levels and the same average degree, $\langle k\rangle=8$. Here $\alpha$ tunes the heterogeneity of the network ($\alpha=0$ corresponds to a SF network, while $\alpha=1$ to an ER network). The curves refer to different values of $s$ ($s=0.1$ blue filled triangles, $s=0.5$ red open triangles, $s=1$ magenta filled squares, $s=2$ cyan open circles, $s=5$ green open squares, $s=10$ black filled cyrcles). All networks have $N=500$ nodes.}
\end{figure}

\section{Critical edges for synchronization}

To further study the interplay between synchronizablity and structure, in this section we investigate the effect of the edge removal on synchronization in the absence and in the presence of LRIs. More specifically, edge
removal is performed according to different strategies; we consider either the removal of links chosen at random or according to some criterion ranking the edges. Since synchronization occurs through the exchange, among the network nodes, of the information on their dynamical state, those links in which information traffic is larger
should be considered as the most critical. To account for this, we considered different measures of edge centrality.

The first one is edge degree calculated as $k_{i}+k_{j}-2$, where $k_{i}$ and $k_{j}$
are the node degrees of $i$ and $j$. The larger is the edge degree
the more critical is the edge, so first we remove edges with larger
edge degree.

Edge degree, however, only considers one-hop information exchanges,
whereas information in a network is transmitted through the many existing
paths from node to node. If one limits to consider shortest paths,
edge centrality can be measured by edge betweeness centrality (EBC)
defined as

\begin{equation}
EBC(e)=\sum\limits _{v_{i}\in V}\sum\limits _{v_{j}\in V}\dfrac{\rho(v_{i},e,v_{j})}{\rho(v_{i},v_{j})},\label{eq:EBC}
\end{equation}

\noindent where $\rho(v_{i},e,v_{j})$ is the number of shortest paths
from node $v_{i}$ to $v_{j}$ passing through $e$ and $\rho(v_{i},v_{j})$
the total number of shortest paths between nodes $v_{i}$ and $v_{j}$.
The larger is the value of the EBC, the more critical is the edge
and so has to be removed first.

If information is assumed to flow not only through shortest paths,
then one can takes into account two other measures of edge centrality
as recently introduced in \cite{estrada2015communicability}: the
communicability function and the communicability angle. The first
is defined as

\begin{equation}
G_{ij}=\sum\limits _{k=0}^{\infty}\dfrac{(\mathrm{A}^{k})_{ij}}{k!}=(e^{\mathrm{A}})_{ij}.
\end{equation}

The communicability function is calculated for the network edges,
that is, $\tilde{G}_{ij}$ where $(i,j)\in E$, providing a measure
to rank them: the smaller is $\tilde{G}_{ij}$ the poorer is the communicability
between $i$ and $j$, so the more critical is the edge. Therefore,
if one wants to remove the critical edges according to this measure,
those edges with small values of $\tilde{G}_{ij}$ should be removed
first.

Finally, the communicability angle is defined as

\begin{equation}
\theta_{ij}=\cos^{-1}\frac{G_{ij}}{\sqrt{G_{ii}G_{jj}}}.
\end{equation}

Restricting the analysis to the network edges, one has $\tilde{\theta}_{ij}$
where $(i,j)\in E$. The larger is $\tilde{\theta}_{ij}$ the poorer
is the communicability between $i$ and $j$, so the more critical
is the edge. Edges with high values of $\tilde{\theta}_{ij}$ should
be then removed first.

Given a graph $G$, for each of the edge centralities considered,
we have ranked the edges in decreasing order and removed a percentage
of those that do not disconnect the graph, creating a graph $G'$ with the same nodes of $G$ and edges $E'=E\backslash \{e\}$.
We have then compared the synchronization measures ($\lambda_{2}/N$
and $Q$) for network $G'$ and that of the
original network $G$. We have then considered LRIs in both $G$ and
$G'$ and again compared the synchronization measures for these
graphs. The analysis has been performed for each of the 54 real-world
networks of the dataset.

Fig.~\ref{fig:realnets2} shows how the synchronizability of real-world
networks, measured by the normalized parameters $\lambda_{2}^*/\lambda_2$ (Fig.
\ref{fig:realnets2}(a)-(c)) and $Q^*/Q$ (Fig.
\ref{fig:realnets2}(d)-(f)), is affected by LRIs (here, $\lambda_{2}^*$ and $Q^*$ indicate the values of $\lambda_{2}$ and $Q$ for the network after removing the edges). To facilitate the
visualization, the networks have been ordered according to the descending
values of $\lambda_2^*/\lambda_2$, where the removal of the links has been done according to decreasing values of the communicability angle.

We observe that the edge removal affects more significantly synchronization
in graphs with LRI than in the no-LRI scheme due to the fact that,
when a ``physical'' link is removed, then also its long-range influence
is removed. Synchronizability of the network without the removed edges
is still larger if LRIs are allowed than if it is not, which means
that the system can work better after the edge removal if such long-range
interactions are present than if not. Finally, comparing the different
edge removal methods, we notice that edge degree does not identify
important edges, as the removal by this index leaves the networks
almost unaffected in terms of the synchronizability. On the contrary,
it seems that the best identifiers of critical edges are those accounting
for effects going beyond nearest neighbors interactions, that is,
the edge betweenness (shortest paths) and the communicability angle
(all walks), because the removal by them affected the most the synchronizability.

\begin{figure}
\centering{}
\subfigure[]{\includegraphics[width=0.32\textwidth]{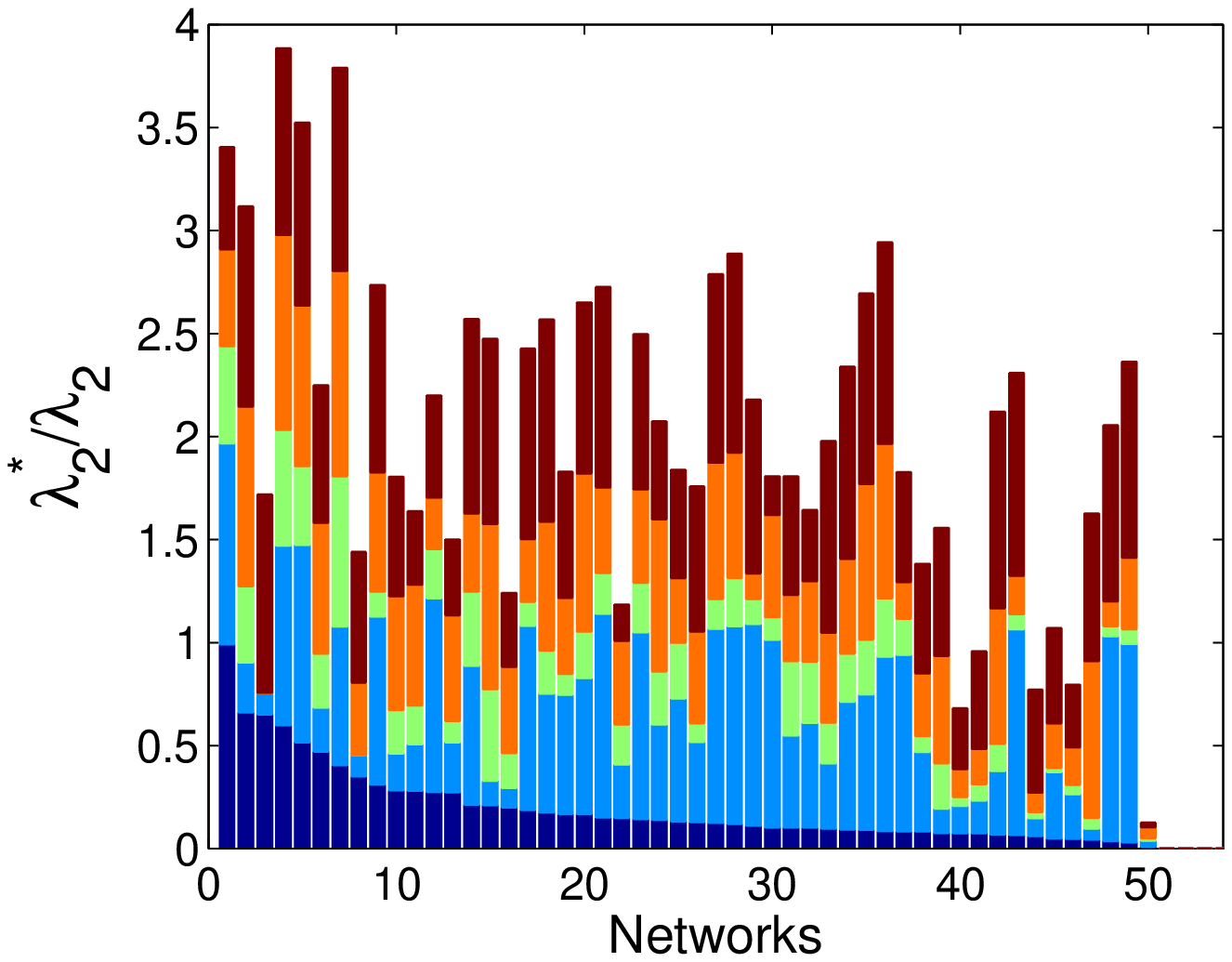}}
\subfigure[]{\includegraphics[width=0.32\textwidth]{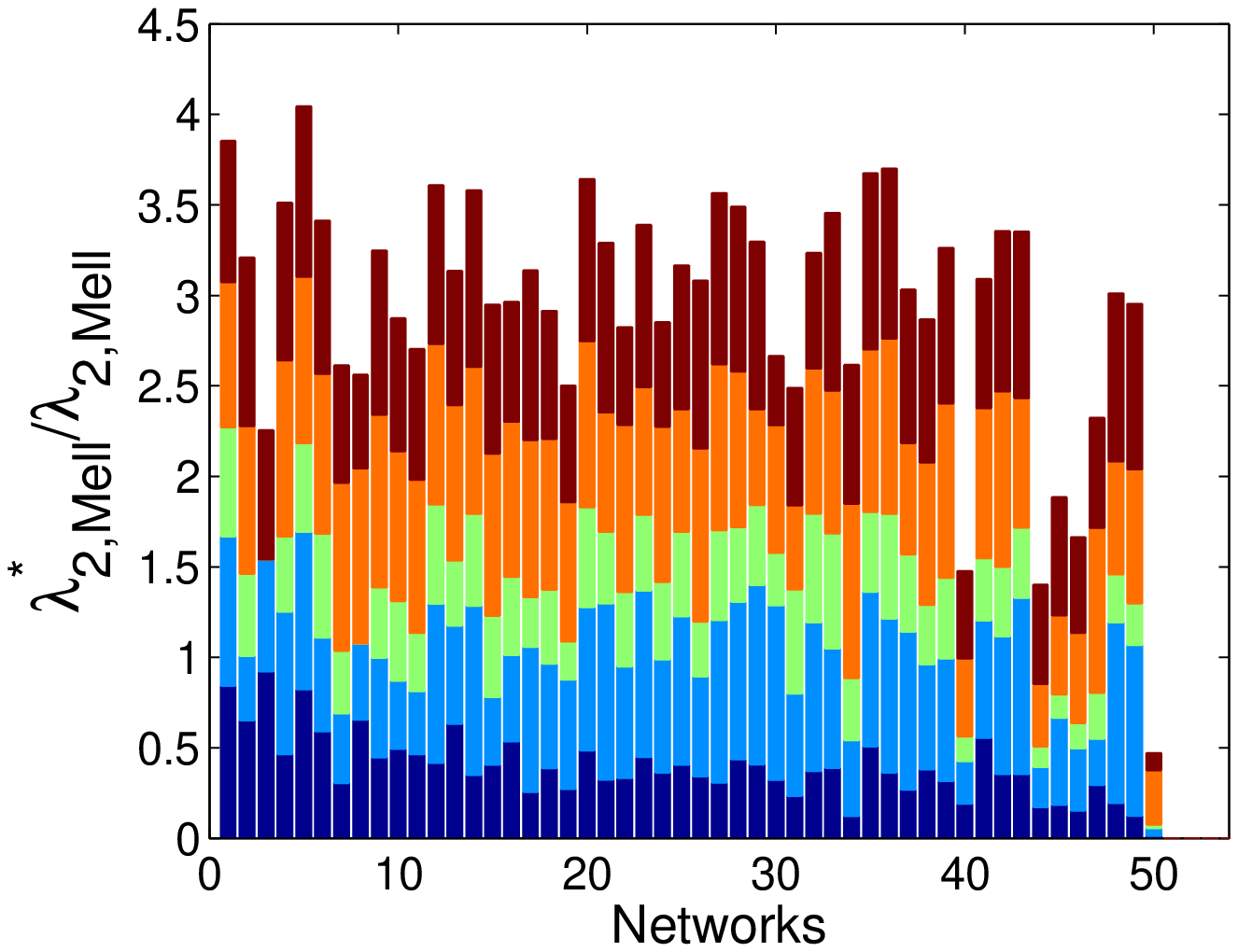}}
\subfigure[]{\includegraphics[width=0.32\textwidth]{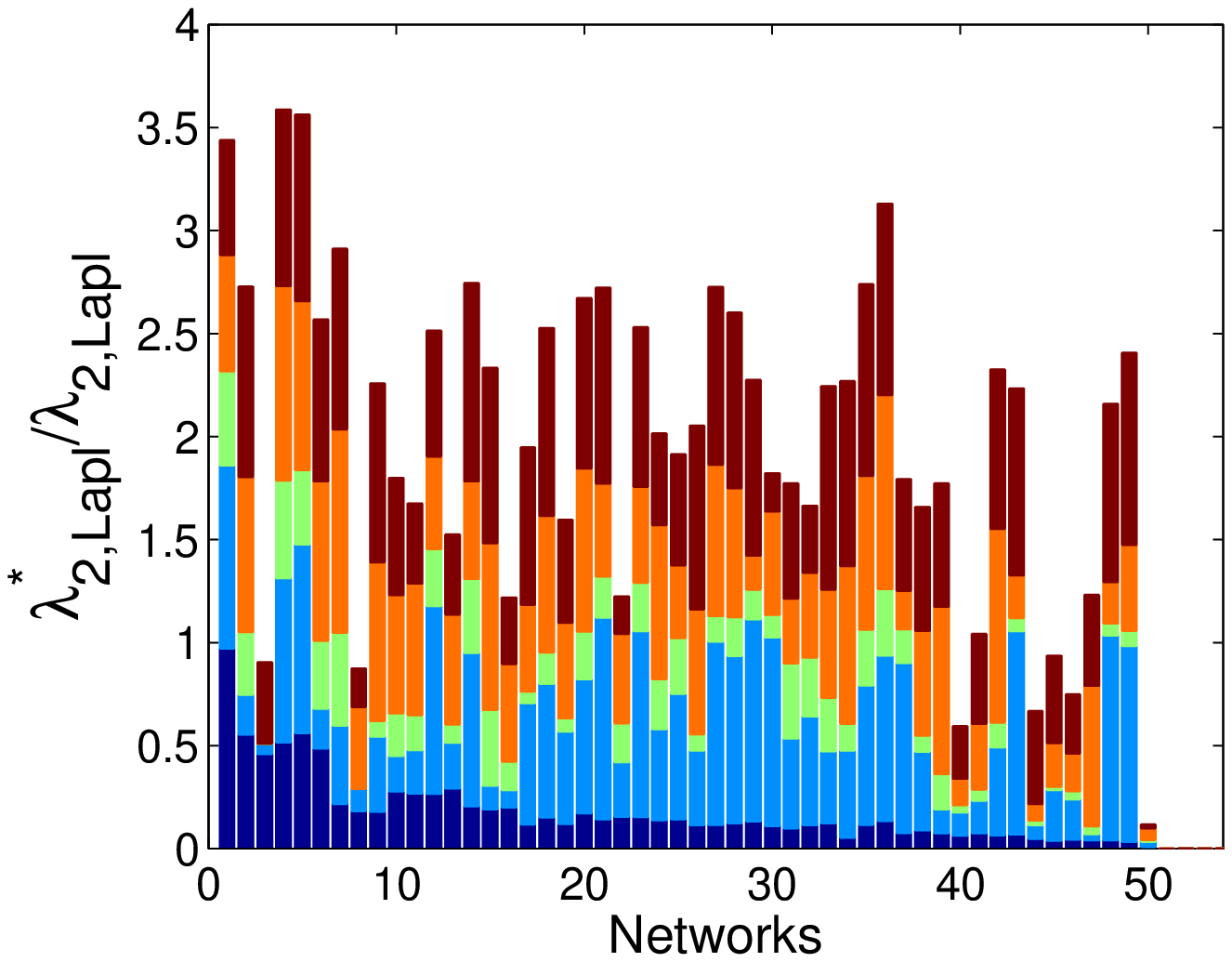}}\\
\subfigure[]{\includegraphics[width=0.32\textwidth]{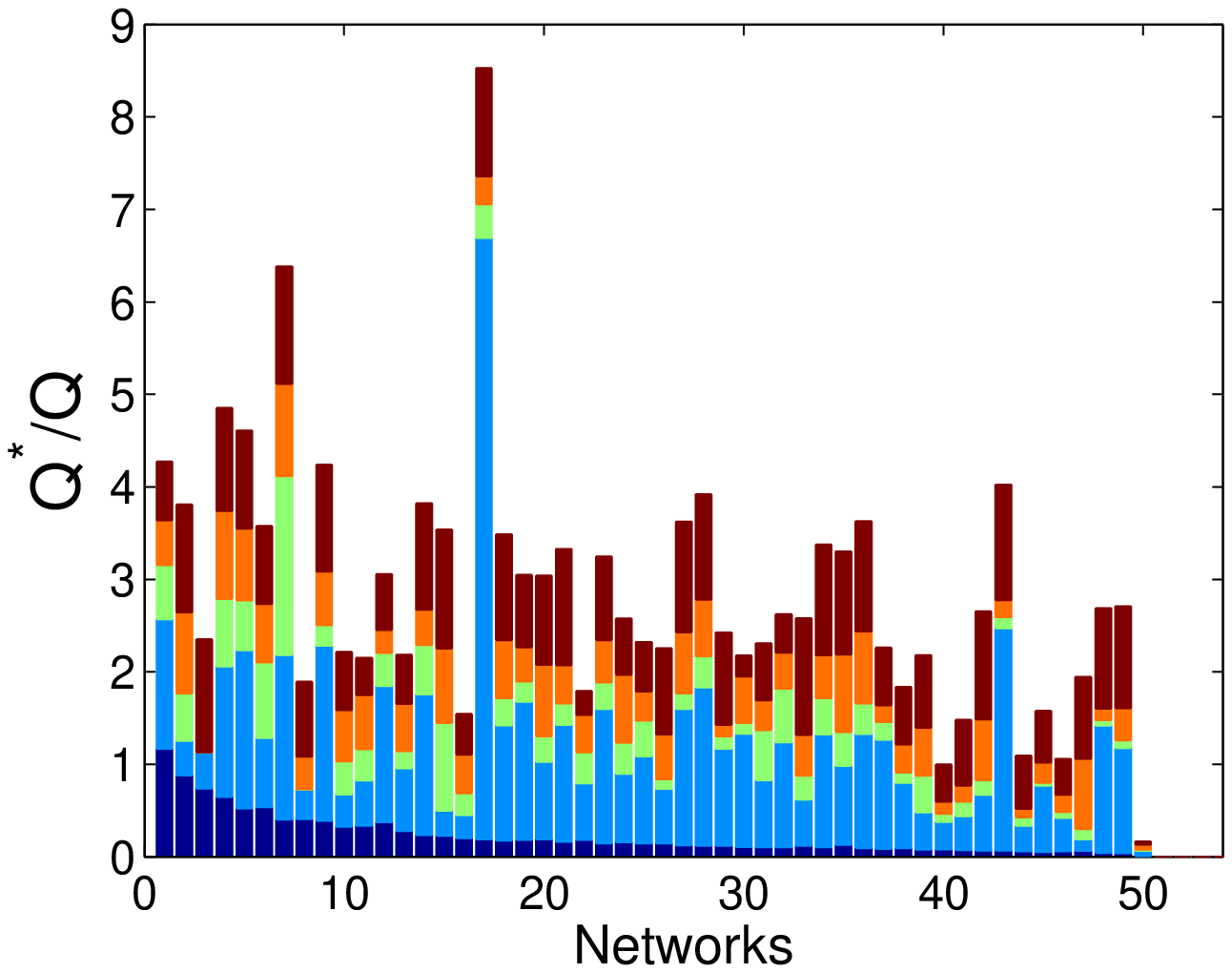}}
\subfigure[]{\includegraphics[width=0.32\textwidth]{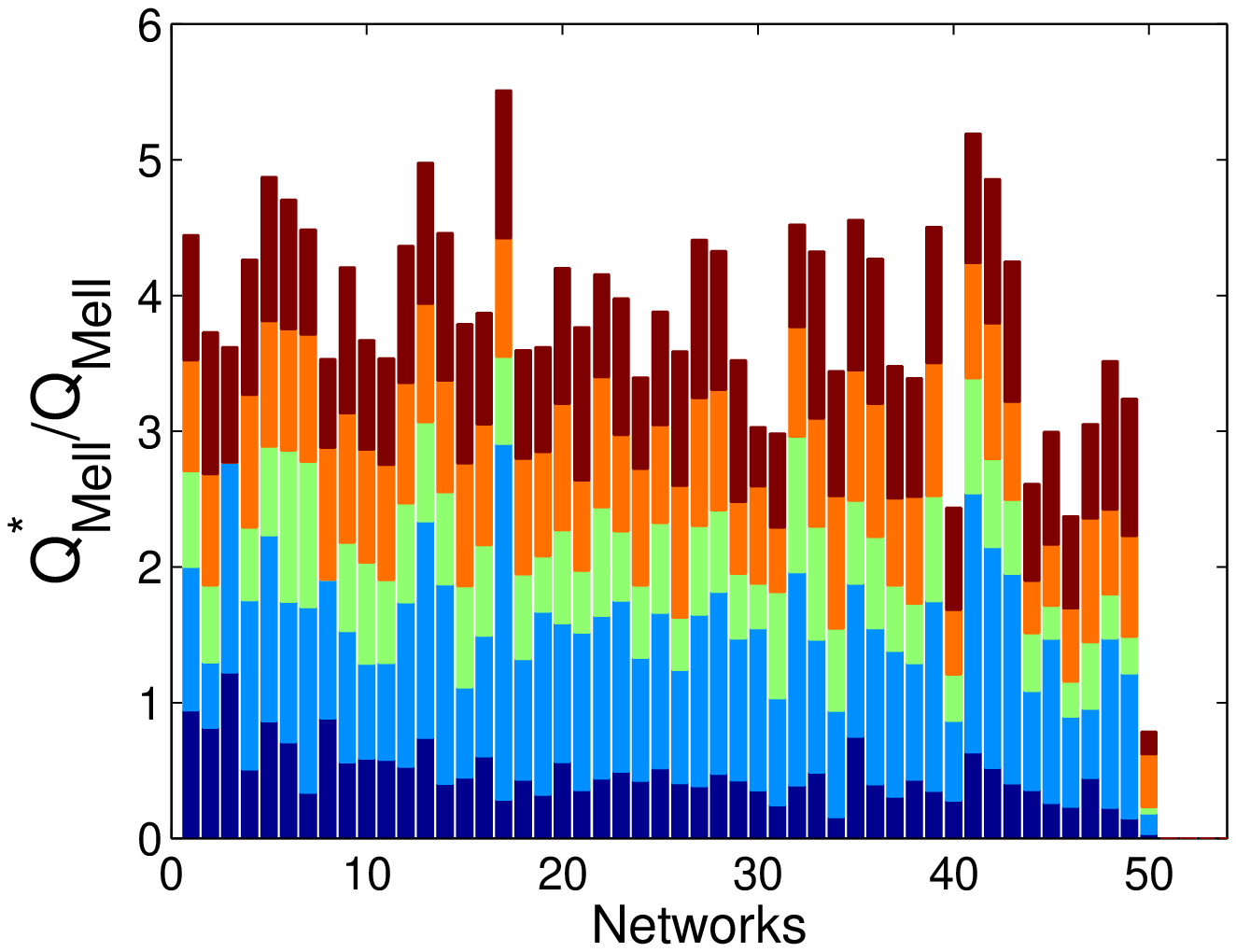}}
\subfigure[]{\includegraphics[width=0.32\textwidth]{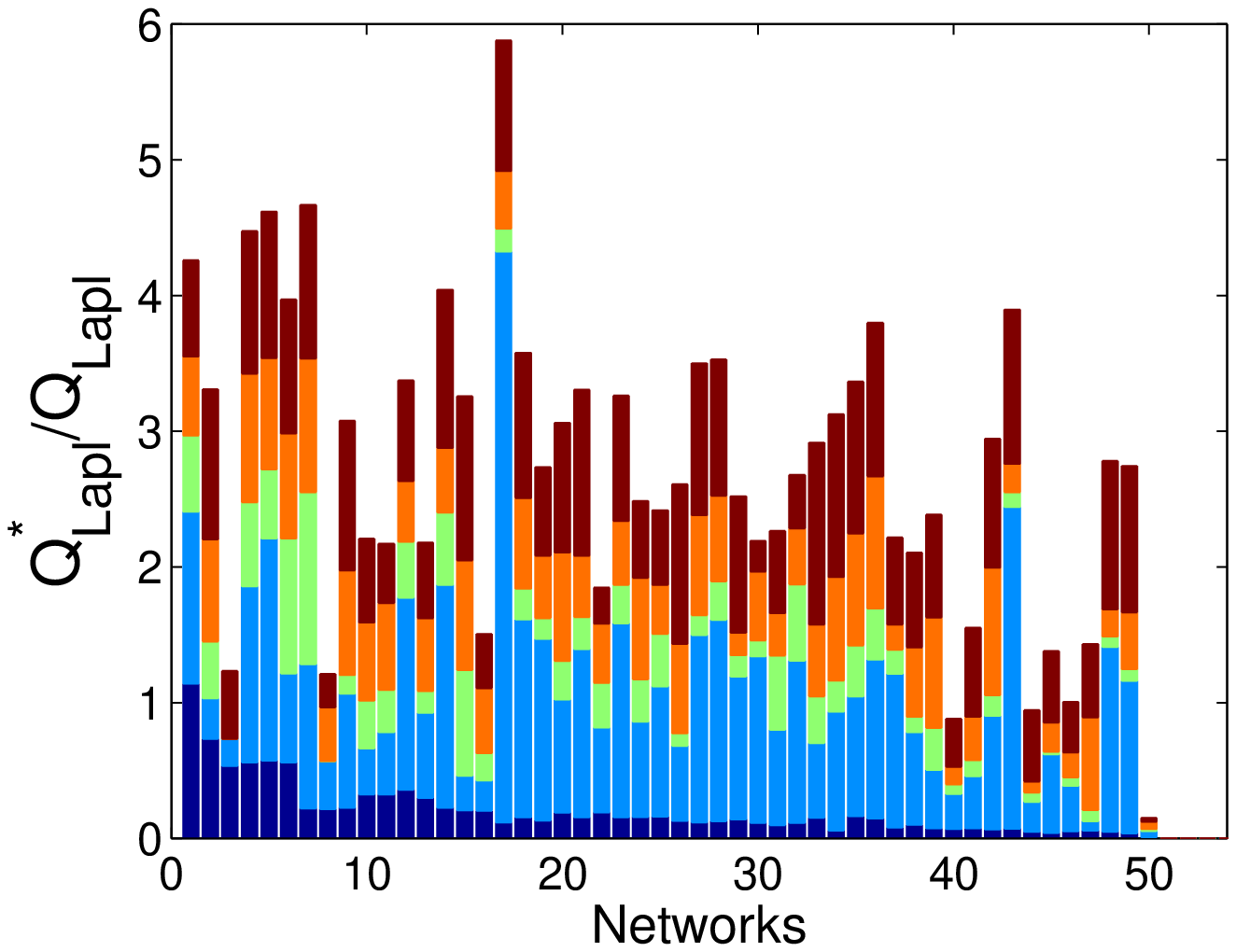}}\\
\caption{\label{fig:realnets2} Effect of edge removal on $\lambda_{2}^*/\lambda_{2}$ (a)-(c) and on $Q/Q^*$ (d)-(f) for real-world networks without LRIs (a) and (d) and with LRIs (b) and (e) (Mellin transform) and (c) and (f) (Laplace transform). The different bars refer to removal according to these edge centralities: from the bottom of each panel to the top, communicability angle (dark blue bars); edge degree (light blue bars); edge betweenness (light green bars); communicability function (orange bars); and random (red bars). The path Laplacians are obtained by applying the Mellin transform or the Laplace transform with $s=2$ or $\lambda=2$. In these results $G'$ has been obtained by removing the 20\% of the original links from $G$.}
\end{figure}

\section{Conclusions}

In this work we have studied the impact of long-range interactions (LRIs) on synchronization. To account for LRIs, oscillators are coupled through $d-$path Laplacian matrices into a model that generalizes the traditional one, based on the classical Laplacian of a graph, and includes it as a special case for $d=1$. As network synchronizability is essentially determined by the spectra of the Laplacian, and in particular by the smallest non-zero eigenvalue or by the ratio between the smallest non-zero eigenvalue and the largest one, we have thus studied these quantities for the graphs with LRIs.

Our theoretical considerations led to the conclusion that, increasing the weight of LRIs, any network of coupled oscillators, independently of the topology and the dynamics of the units, approaches the best possible scenario for synchronization. The specific path towards this limit condition depends on the network properties. We have performed numerical simulations both on real-world examples and on network models, and found results in perfect agreement with the theoretical expectations.

A significant result is the dependence of the impact of LRIs on some topological properties of interest, such as the diameter, the density and the degree distribution. We have found that a larger diameter, a smaller average degree or a higher heterogeneity in the degree distribution are all factors contributing to increase the positive impact of LRIs on synchronizability with respect to the scenario without LRIs.

Finally, we have studied the effects of the removal of critical edges in the absence and in the presence of LRIs, where criticality has been measured according to different topological measures of edge centrality. This analysis, carried out on real-world networks of different sizes and characteristics, points out scenarios common to all the examples considered. First of all, we have found that edge removal has a larger impact on the synchronizability of networks with LRIs rather than on the structures without LRIs. This is due to the fact that, when a physical link is removed, then also the long-range interactions it allowed are affected. However, synchronizability of the networks without the removed edges is still larger if LRIs are allowed than if not, confirming the general result that the presence of long-range interactions always favors synchronizability. Finally, we have observed that the most critical edges for synchronization are those having high values of the edge betweenness or of the communicability angles, i.e., are ranked according to measures taking into account effects going beyond nearest neighbors interactions.

The current approach can be extended to directed graphs. The first step in doing so is to generalize the $d$-path Laplacian matrices for such graphs. In this case there are two kinds of $d$-path Laplacians, namely the in- and the out-degree $d$-path Laplacians. That is, we should first generalize the in- and out-degrees to account for the indirect influence of nodes. The generalization is, however, straightforward. We only need to consider the number of directed shortest paths of length $d$ from the node $i$ to any other node in the graph to account for the out-$d$-degree of the node $i$. Similarly, we can define the in-$d$-degree of the node $i$ by taking the number of directed shortest paths of length $d$ from any node in the graph to the node $i$. For the synchronization dynamics we should consider the out-degree $d$-path Laplacians in a similar way as we have done in the current work for the undirected graphs. For directed graphs synchronizability has to be checked in the complex plane as the Laplacian eigenvalues are in general complex. However, provided that the graph is strongly connected, the generalization of $d$-path Laplacian matrix accounting for oriented paths still yields for $s\rightarrow 0$ ($\lambda \rightarrow 0$) a complete graph, thus favoring synchronizability. The mathematical properties of these in- and the out-degree $d$-path Laplacians are not trivial and deserve to be considered in a separate work.

\appendix
\section{Real-world network dataset}

The real-world networks used in this paper belong to different domains: ecological (includes food webs and ecosystems), social (networks of friendships, communication networks, corporate relationships), technological (internet, transport, software development networks), informational (vocabulary networks, citations) and biological (protein-protein interaction networks, transcriptional regulation networks). The dataset comprises networks of different sizes, ranging from $N=29$ to $N=4941$ nodes. The networks are listed in Table~\ref{tab:tab1}.

\begin{table}[ht]
\caption{Dataset of real-world networks: network name, domain, $N$ number of nodes, $m$ number of links, and reference. The networks have been ordered according to the descending values of $\lambda_{2}/N$.} 
\centering
\begin{footnotesize}
\begin{tabular}{c l c c c c c}
\hline\hline
No. & Dataset & Domain & N & m & $\lambda_2^*/\lambda_2$ & Ref.\\ [0.5ex]
\hline
1 & ReefSmall & ecological & 50 & 524 & 0.9904 & \cite{refn1}\\
2 & StMartin & ecological & 44 & 218 & 0.6589 &\cite{refn2}\\
3 & Internet-1998 & technological & 3522 & 6324 & 0.6503 &\cite{refn8}\\
4 & Trans-urchin & biological & 45 & 80 & 0.5974 &\cite{refn4}\\
5 & KSHV & biological & 50 & 122 & 0.5143 &\cite{refn5}\\
6 & ODLIS & informational & 2898 & 16381 & 0.4696 &\cite{refn9}\\
7 & Software Abi & technological & 1035 & 1736 & 0.4036 &\cite{refn7}\\
8 & Internet-1997 & technological & 3015 & 5156 & 0.3488 &\cite{refn8}\\
9 & USA Air 97 & technological & 332 & 2126 & 0.3095 &\cite{refn9}\\
10 & Sawmill & social & 36 & 62 & 0.2802 &\cite{refn10}\\
11 & Grassland & ecological & 75 & 113 & 0.2789 &\cite{refn11}\\
12 & World-trade & informational & 80 & 875 & 0.2726 &\cite{batagelj2006analysis} \\
13 & Trans-Ecoli & biological & 328 & 456 & 0.2715 &\cite{refn4}\\
14 & Neurons & biological & 280 & 1973 & 0.2122&\cite{refn14}\\
15 & Ythan1 & ecological & 134 & 597 & 0.2102&\cite{refn15}\\
16 & Hpyroli & biological & 710 & 1396 & 0.1978&\cite{refn16}\\
17 & Software Mysql & technological & 1480 & 4221 & 0.1850&\cite{refn7}\\
18 & YeastS & biological & 2224 & 7049 & 0.1736&\cite{refn18}\\
19 & Geom & social & 3621 & 9461 & 0.1666&\cite{batagelj2006analysis} \\
20 & Chesapeake & ecological & 33 & 72 & 0.1664&\cite{refn20}\\
21 & Benguela & ecological & 29 & 191 & 0.1490&\cite{refn21}\\
22 & Software Digital & technological & 150 & 198 & 0.1470&\cite{refn7}\\
23 & Hi-tech & social & 33 & 91 & 0.1420&\cite{refn23}\\
24 & Zackar & social & 34 & 78 & 0.1375&\cite{refn24}\\
25 & PRISON-Sym & social & 67 & 142 & 0.1309&\cite{refn25}\\
26 & ScotchBroom & ecological & 154 & 370 & 0.1269&\cite{refn26}\\
27 & PIN Ecoli & biological & 230 & 695 & 0.1231&\cite{refn27}\\
28 & Roget & informational & 994 & 3641 & 0.1171&\cite{refn28}\\
29 & Skipwith & ecological & 35 & 364 & 0.1098&\cite{refn29}\\
30 & StMarks & ecological & 48 & 221 & 0.1011&\cite{refn30}\\
31 & social3 & social & 32 & 80 & 0.1002&\cite{refn31}\\
32 & Malaria PIN & biological & 229 & 604 & 0.0999&\cite{refn32}\\
33 & Canton & ecological & 108 & 708 & 0.0943&\cite{refn33}\\
34 & Drugs & social & 616 & 2012 & 0.0917&\cite{moody01} \\
35 & BridgeBrook & ecological & 75 & 547 & 0.0896&\cite{refn35}\\
36 & Stony & ecological & 112 & 832 & 0.0844&\cite{refn36}\\
37 & Dolphins & social & 62 & 159 & 0.0826&\cite{refn37}\\
38 & ElVerde & ecological & 156 & 1441 & 0.0813&\cite{refn38}\\
39 & PIN Human & biological & 2783 & 6438 & 0.0733&\cite{rual2005towards} \\
40 & electronic2 & technological & 252 & 399 & 0.0730&\cite{refn40}\\
41 & Software-VTK & technological & 771 & 1369 & 0.0717&\cite{refn7}\\
42 & Transc-yeast & biological & 662 & 1062 & 0.0668&\cite{refn4}\\
43 & CorporatePeople & social & 1586 & 13126 & 0.0654&\cite{refn43}\\
44 & electronic3 & technological & 512 & 819 & 0.0595&\cite{refn40}\\
45 & Software-XMMS & technological & 971 & 1809 & 0.0471&\cite{refn7}\\
46 & electronic1 & technological & 122 & 189 & 0.0463&\cite{refn40}\\
47 & SmallW & informational & 233 & 994 & 0.0422&\cite{refn9}\\
48 & Shelf & ecological & 81 & 1476 & 0.0356&\cite{link2002does} \\
49 & Coachella & ecological & 30 & 261 & 0.0300&\cite{refn49}\\
50 & Power grid & technological & 4941 & 6594 & 0.0067&\cite{watts1998collective} \\
51 & ColoSPG & social & 324 & 347 & $1.8\cdot 10^{-13}$ &\cite{refn51}\\
52 & Drosophila PIN & biological & 3039 & 3715 & $3\cdot 10^{-14}$&\cite{giot2003protein} \\
53 & Pin- Bsubtilis & biological & 84 & 98 & $1.7\cdot 10^{-15}$&\cite{refn53}\\
54 & PIN-Afulgidus & biological & 32 & 38 & $1.4\cdot 10^{-15}$&\cite{refn54}\\ [1ex] 
\hline 
\end{tabular}
\end{footnotesize}
\label{tab:tab1} 
\end{table}


\bibliographystyle{plain}
\bibliography{MultiHopK}

\end{document}